\documentclass[journal]{IEEEtran}

\usepackage{amssymb}
\usepackage{amsthm}
\usepackage{amsmath}
\usepackage{mathtools}
\usepackage{caption}
\usepackage{epsfig} 
\usepackage{subfigure}
\usepackage{times} 
\usepackage{amsmath} 
\usepackage{amssymb}  
\usepackage{soul} 
\usepackage{pdfsync}
\usepackage{enumerate}
\usepackage{color}
\usepackage{lineno}
\usepackage{amsmath,bm}
\usepackage{cite}
\usepackage[colorlinks,linkcolor=red,anchorcolor=blue,citecolor=blue]{hyperref}

\usepackage[dvipsnames]{xcolor} 

\usepackage{dsfont} 

\usepackage{graphicx}
	\graphicspath{{figures/}}

\usepackage{tikz}
	\usetikzlibrary{quotes,angles} 
	\tikzstyle{frame} = [draw, -latex]
	\tikzstyle{lineUD} = [draw]
	\tikzstyle{line} = [draw, -latex']
	\tikzstyle{line2} = [draw, -latex', dashdotted]
	\tikzstyle{line3} = [draw, -latex', dashed]
	\tikzstyle{line3UD} = [draw, dashed]
	\tikzstyle{place} = [circle, draw=black, fill=white, thick, inner sep=2pt, minimum size=1mm]
	\tikzstyle{placeRed} = [circle, draw=red, fill=red, thick, inner sep=2pt, minimum size=1mm]
	\tikzstyle{vertex} = [circle, draw=black, fill=black, thick, inner sep=2pt, minimum size=1mm]
\usepackage{tkz-euclide}

\newtheorem{definition}{Definition}[section]
\newtheorem{problem}{Problem}[section]

\newtheorem{remark}{Remark}[section]
\newtheorem{lemma}{Lemma}[section]
\newtheorem{theorem}{Theorem}[section]
\newtheorem{corollary}{Corollary}[section]

\newtheorem{proposition}{Proposition}

\newcommand{\mbf}{\mathbf}

\newcommand{\myemph}{\emph}
\newcommand{\myfig}{Fig.~}
\DeclareMathOperator{\rank}{{rank}}
\DeclareMathOperator{\vspan}{{span}}

\DeclareMathOperator{\vnull}{{Null}}
\DeclarePairedDelimiter \card{\lvert}{\rvert}
\DeclarePairedDelimiter \norm{\|}{\|}

\providecommand\given{}
\newcommand\SetSymbol[1][]{\nonscript\;#1\vert\nonscript\;
\mathopen{}\allowbreak}
\DeclarePairedDelimiterX\set[1]\{\}{%
\renewcommand\given{\SetSymbol[\delimsize]}
#1
}

\title{\LARGE \bf
Hybrid distance-angle rigidity theory with signed constraints \\ and its applications to formation shape control
}

\author{Seong-Ho Kwon, Zhiyong Sun,
        Brian D. O. Anderson, and Hyo-Sung Ahn
\thanks{S.-H. Kwon and H.-S. Ahn are with the School of Mechanical Engineering, Gwangju Institute of Science and Technology, Gwangju 61005, Korea (e-mail: $\{$seongho; hyosung$\}$@gist.ac.kr).}
\thanks{Z. Sun is with Department of Automatic Control, Lund University, Sweden (e-mail:  sun.zhiyong.cn@gmail.com,
zhiyong.sun@control.lth.se).}
\thanks{B. D. O. Anderson is with Research School of Engineering, The Australian National University, Canberra, ACT 2601, Australia (e-mail: Brian.Anderson@anu.edu.au).}
}

\begin{document}
\pagestyle{plain}
\maketitle

\begin{abstract}
In this paper, we develop a hybrid distance-angle rigidity theory that involves  heterogeneous distances (or unsigned angles) and signed constraints for a framework in the 2-D and 3-D space. The new rigidity theory determines a (locally) unique formation shape up to a translation and a rotation by a set of distance and signed constraints, or up to a translation, a rotation and additionally a scaling factor by a set of unsigned angle and signed constraints. Under this new rigidity theory,  we have a clue to resolve the flip (or reflection) and flex ambiguity for a target formation with hybrid distance-angle constraints. In particular, we can completely eliminate the ambiguity issues if formations are under a specific construction which is called \myemph{signed Henneberg construction} in this paper. We then apply the rigidity theory to formation shape control and develop a gradient-based control system that guarantees an exponential convergence close to a desired formation by inter-neighbor measurements. Several numerical simulations on formation shape control with hybrid distance-angle constraints are provided  to validate the theoretical results. 
\end{abstract}


\section{Introduction}
\label{Sec:intro}
There has been active research on multi-agent formation control and network localization via graph rigidity theories in the literature\cite{sun2016exponential,zhao2016bearing,jing2018weak, kwon2019generalized}, where the rigidity theories have particularly been a significant part in characterizing unique formation shapes and formation controller design. Briefly speaking, we can classify the rigidity theories according to different types of constraints: the rigidity theory based on distance constraints is traditionally well known and is termed \myemph{(distance) rigidity theory} \cite{asimow1979rigidity,roth1981rigid,C:Hendrickson:SIAM1992}; the rigidity theory with bearing constraints is called  \myemph{bearing rigidity theory} \cite{zhao2016bearing,zhao2017translational}; the rigidity theory with a set of distance and angle constraints has been studied in \cite{jing2018weak, kwon2018infinitesimal_inp.,kwon2018generalized,bishop2015distributed} and is termed \myemph{weak rigidity theory}; the authors in \cite{jing2019angle,kwon2018generalized,chen2019angle} introduce the rigidity theory with only angle constraints, where the works in \cite{kwon2018generalized,jing2019angle} are based on relative position measurements and the work \cite{chen2019angle} introduces \myemph{angle rigidity theory} that considers formation control under only angle measurements.
In recent years, generalized rigidity theories have been studied in \cite{stacey2017role,kwon2019generalized}, where the authors in \cite{stacey2017role} proposed a generalized framework of rigidity theory with heterogeneous states of all agents on non-Euclidean spaces and general relative state constraints, and the authors in \cite{kwon2019generalized} presented the generalized rigidity with heterogeneous constraints. Furthermore,  bearing-ratio-of-distance rigidity theory  is introduced in \cite{cao2019bearing} by using bearing and ratio-of-distance constraints.
 Lastly, clique rigidity theory has been studied in \cite{sakurama2018multi}, where a formation is decomposed into maximal cliques to check whether the formation is unique or not.

As is well known, formation control based on the (distance) rigidity theory has widely been studied in the literature \cite{krick2009stabilisation,cortes2009global,sun2014finite,sun2016exponential}. 
However, although the (distance) rigidity theory has been a powerful tool in the field of formation control, there have always been challenging tasks of resolving shape ambiguities such as \myemph{flip ambiguity} and \myemph{flex ambiguity}\footnote{The notion of flex is used in a different sense from that of flex/flexible of a framework. That is, a framework which can cause a flex ambiguity does not mean that it is not rigid.} \cite{anderson2008rigid}. 
The main reason why these issues occur is that the (distance) rigidity theory cannot distinguish formation shapes under flip and flex transformations with only distance constraints. These ambiguity phenomena cannot guarantee a convergence of formations to a desired shape specified by only inter-agent distances.
The ambiguity issues also remain in  formation controls based on the weak rigidity and angle rigidity theories \cite{jing2018weak, kwon2018infinitesimal_inp.,kwon2018generalized,bishop2015distributed,jing2019angle}. 

To solve the flip and flex ambiguity in formation control, there have been several recent efforts in the literature. 
In terms of the bearing rigidity theory, the flip and flex ambiguity is not relevant due to the vector bearing information under a common coordinate system, but the bearing-based formation control demands additional implementation requirements such as orientation alignment in local coordinate frames, relative orientation measurement, etc.
In the work \cite{chen2019angle}, the authors introduce an angle-based formation control with signed angles; however, they only provide the 2-dimensional special results confined to specific formation shapes.
The authors in \cite{sun2017distributed} consider not only  formation control with distance constraints but also orientation control, but the proposed controllers also require a global coordinate system or orientation alignments of some agents. 
Another recent paper \cite{anderson2017formation} utilizes signed area constraints in dealing with flip and flex ambiguity, but the work only considers special cases of 3- and 4-agent formation control in the 2-dimensional space.

In this paper, we will focus on the problem of  analyzing general formations without a global coordinate system and orientation alignments in a distributed way. Moreover, we use heterogeneous scalar constraints involving geometric functions such as distance, angle and signed constraints. Heterogeneous constraints are ubiquitous in the field of multi-agent coordination, partly motivated by the application of   different sensing information and capabilities among individual agents, and distinct geometric variables in characterizing a coordination task  \cite{santos2018coverage, liu2019further}. To address the formation shape control problem with heterogeneous distance and angle constraints,  we will develop a new theory on graph rigidity. The new theory can determine a locally unique formation shape\footnote{A locally unique formation shape means that the shape cannot be deformed by an infinitesimal smooth motion of the formation except a translation and a rotation of entire formation by a set of distance and signed constraints (or a translation, a rotation and a scaling of entire formation by a set of unsigned angle and signed constraints).} and provide a clue to partially resolve  the flip and flex ambiguity in formation control. 
In particular, if formations are under a specific construction, which is termed \myemph{signed Henneberg construction} in this paper, then we can completely eliminate the ambiguity issues. The proposed framework and the new rigidity theory are applicable to a wide area of multi-agent system tasks, e.g., network localization and formation control. As an application example, we explore formation control design with heterogeneous scalar constraints based on the proposed hybrid distance-angle rigidity theory and then show a locally exponential convergence of target formations without leading to ambiguity phenomena.

We summarize the main results of our work as follows. We firstly propose a \myemph{hybrid distance-angle rigidity theory} with signed constraints to address the flip and flex ambiguity; the hybrid distance-angle rigidity includes two novel concepts, namely, \myemph{infinitesimal rigidity with distance and signed constraints (IRDS)} and \myemph{infinitesimal rigidity with angle and signed constraints (IRAS)}. The new rigidity theory is used to characterize a (locally) unique formation up to a translation and a rotation for the IRDS, and up to a translation, a rotation and a scaling factor for the IRAS. 
To completely eliminate the flip and flex ambiguity, we consider and develop a sequential technique termed \myemph{signed Henneberg construction}.
Then, by employing the proposed hybrid rigidity theory, we design formation control systems with a gradient flow law in the 2- and 3-dimension space, 
where the systems guarantee a local convergence to a target formation. Moreover, the controllers do not require a global coordinate system, that is, all individual agents only use local coordinate systems and inter-neighbor relative measurements in implementing the distributed formation control law. 

The remaining parts of the paper are organized as follows. Section \ref{Sec:Pre_Not} provides notations and preliminaries. Then, Section \ref{Sec:inf_rigid} proposes new concepts to characterize a unique formation with a set of distance (or angle) and signed constraints. With the new concepts, we propose control laws in Section \ref{Sec:controller}. Finally, we provide numerical simulations and conclusions in Section \ref{Sec:simul}  and Section \ref{Sec:conclusions}, respectively. 


\section{Preliminaries and motivation}
\label{Sec:Pre_Not}
This section briefly reviews formation control laws using \myemph{infinitesimal (distance) rigidity} \cite{krick2009stabilisation,cortes2009global,sun2014finite,sun2016exponential} and  \myemph{infinitesimal weak rigidity} \cite{kwon2018infinitesimal_inp.,kwon2018generalized}. 
We then explain  the flip and flex ambiguity   and discuss how we eliminate these formation ambiguities to motivate the new rigidity concepts.

\subsection{Formation graph and notations} \label{Inf_rigid_pre}
We firstly introduce preliminaries and notations  to characterize formations with graphs.
The symbols $\norm{v}$ and $\card{\mathcal{S}}$ denote the Euclidean norm of a vector $v \in \mathbb{R}^{d}$ and the cardinality of a set $\mathcal{S}$, respectively. The symbols $\vnull(\cdot)$ and $\rank (\cdot)$ denote the null space and rank of a matrix, respectively.
Then, undirected graph $\mathcal{G}$ is denoted by $\mathcal{G} = (\mathcal{V},\mathcal{E}, \mathcal{A})$, where the symbols $\mathcal{V}$, $\mathcal{E}$ and $\mathcal{A}$ denote a vertex set, an edge set and an angle set, respectively. In the graph $\mathcal{G}$, it is defined that $\mathcal{V}=\set{1,2,...,n}$, $\mathcal{E} \subseteq \mathcal{V} \times \mathcal{V}$ with $m_d=\card{\mathcal{E}}$ and 
{\color{black} $\mathcal{A} = \set{(i,j,k) \in \mathcal{V}^3 \given i,j,k \in \mathcal{V} \text{  to characterize } \theta_{jk}^{i} \in [0,\pi]}$ with $m_a=\card{\mathcal{A}}$}
where $\theta_{jk}^{i}\in[0,\pi]$ denotes a subtended angle by a pair of two edges $(i,j)$ and $(i,k)$. 
 Since we only consider undirected graphs, it is assumed that $(i,j) = (j,i)$ for all $i,j \in \mathcal{V}$. The oriented incidence matrix\footnote{The oriented incidence matrix of an undirected graph is the incidence matrix of arbitrarily chosen orientation of the graph without violating the undirected graph assumption.} $H'$ is defined as $H'=[h'_{ij}]$, where $h'_{ij}=-1$ if $i$-th edge  leaves vertex $j$, $h'_{ij}=1$ if $i$-th edge sinks at vertex $j$, and $h'_{ij}=0$ otherwise.
A realization $p$ is then defined as $p = [p_{1}^\top,...,p_{n}^\top]^\top \in \mathbb{R}^{dn}$ where $p_i \in \mathbb{R}^{d}$ denotes a position vector for each $i \in \mathcal{V}$. Moreover, we can define a \myemph{framework} as $(\mathcal{G},p)$ associated with an undirected graph $\mathcal{G}$ and a realization $p$. In this paper, a framework is equivalently called a formation, and we assume $d=2,3$.

We here present some definitions and notations on graphs and formations which will frequently be used  in this paper. 
We define the relative position vector $z_{ij}$ as $z_{ij} \triangleq p_{i} - p_{j}$, and order the relative position vectors such that $z_g \triangleq z_{ij}$ for $(i,j) \in \mathcal{E}, \forall g\in \{ 1,..., m_d\}$. Similarly, we define the ordered cosines as $A_h \triangleq \cos{\theta_{jk}^{i}}$ for $(i,j,k) \in \mathcal{A},\forall h\in \{ 1,..., m_a\}$. To denote distance constraints, we use the notation $d_{ij}$ defined as $d_{ij} \triangleq \norm{z_{ij}}$ for $(i,j)\in \mathcal{E}, i \neq j$. We denote the orthogonal projection matrix as $P_x=I_d - \frac{x x^\top}{\norm{x}^2}$ for nonzero vector $x\in\mathbb{R}^{d}$.

\subsection{Gradient formation control based on infinitesimal (distance) rigidity and  infinitesimal weak rigidity} \label{Inf_weak_rigid_pre}
In this subsection, we review formation control systems with a gradient flow law \cite{sakurama2015distributed,krick2009stabilisation,bishop2015distributed,park2014stability} based on the concepts of infinitesimal (distance) rigidity and  infinitesimal weak rigidity. 

Let us first consider the following distance and angle constraints. 
\begin{align}
\bold{d}(p) &= \left[\ldots, \frac{1}{2}\norm{z_g}^2, \ldots \right]^\top, \quad g\in \{ 1,..., m_d\}, \\
\bold{a}(p) &= \left[\ldots, A_h, \ldots \right]^\top, \quad h\in \{ 1,..., m_a\}.
\end{align}
Also, the formation errors are defined as
\begin{align}
e_d(p)&=\left[\bold{d}(p)^\top \, - \bold{d}^{*\top} \right]^\top,\\
e_w(p)&=\left[2\bold{d}(p)^\top, \bold{a}(p)^\top \right]^\top - \left[2\bold{d}^{*\top}, \bold{a}^{*\top} \right]^\top.
\end{align}
where $\bold{d}^*$ and $\bold{a}^*$ denote two vectors constituting desired distance constraints $\frac{1}{2}\norm{z^*_g}^2$ of $\frac{1}{2}\norm{z_g}^2$ and desired angle constraints $A^*_h$ of $A_h$, respectively, as follows
\begin{align} 
\bold{d}^* &= \left[\ldots, \frac{1}{2}\norm{z^*_g}^2, \ldots \right]^\top,\\
\bold{a}^* &= \left[\ldots, A^*_h, \ldots \right]^\top.
\end{align}
Then, the formation control systems are represented by
\begin{align}
\dot{p}=u &\triangleq -\left(\nabla \frac{1}{2}\norm{{e_d}(p)}^2 \right)^\top=-R_D^\top(p) e_d(p) \quad \label{control_law_origin01}
\end{align}
or
\begin{align}
\dot{p}=u &\triangleq -\left(\nabla \frac{1}{2}\norm{{e_w}(p)}^2 \right)^\top= -R_W^\top(p) e_w(p).\label{control_law_origin02}
\end{align}
where $R_D$ and $R_W$ denote the \myemph{rigidity matrix} (see e.g., \cite{asimow1979rigidity,roth1981rigid,C:Hendrickson:SIAM1992}) and \myemph{weak rigidity matrix} (e.g., \cite{kwon2018infinitesimal_inp.,kwon2018generalized}), respectively. Note that in this paper we assume that the dynamics of each agent $i$ for $i\in \mathcal{V}$ is modeled by a single integrator. 
The equation \eqref{control_law_origin01} has been a typical control law using the (distance) rigidity theory with distance constraints. Here, the rigidity matrix plays an important role in determining whether or not a formation shape is unique up to a rigid-body translation and a rigid-body rotation of an entire formation.
The following theorem shows the necessary and sufficient condition for the infinitesimal (distance) rigidity determined by the rank of the associated rigidity matrix.
\begin{theorem}[\cite{asimow1979rigidity,C:Hendrickson:SIAM1992}]
A framework $(\mathcal{G},p)$ with $n \ge d$ is infinitesimally (distance) rigid in $\mathbb{R}^{d}$ if and only if  $\rank(R_D(p))=dn-d(d + 1)/2$.  
\end{theorem}

On the other hand, several works in \cite{kwon2018infinitesimal_inp.,kwon2018generalized} recently introduced another rigidity concept termed infinitesimal weak rigidity by using a set of distance constraints and additional  angle constraints as follows.
\begin{theorem}[\cite{kwon2018infinitesimal_inp.,kwon2018generalized}]
A framework $(\mathcal{G},p)$ with $n \ge 3$ is infinitesimally weakly rigid in $\mathbb{R}^{d}$ if and only if  $\rank(R_W(p))=dn-d(d + 1)/2$ (resp. $\rank(R_W(p))=dn-(d^2+d + 2)/2$) in the case of $\mathcal{E} \neq \emptyset$ (resp. $\mathcal{E} = \emptyset$).   
\end{theorem}

Compared to the concept of the infinitesimal (distance) rigidity, infinitesimal weak rigidity displays a particular characteristic in the case of $\mathcal{E} = \emptyset$ where the equality $\mathcal{E} = \emptyset$ means that a formation has no   distance constraints, but only has angle constraints. When $\mathcal{E} = \emptyset$, a formation based on the infinitesimal weak rigidity could be uniquely determined up to a rigid-body translation, a rigid-body rotation and additionally a scaling factor of an entire formation. That is,  the entire formation has one more degree of freedom in the case of $\mathcal{E} = \emptyset$. 

The two rigidity theories reviewed above however cannot address the flip and flex ambiguity of a target formation shape, which motivates us to define a new rigidity theory under hybrid distance and angle constraints. 
 In the next subsection, we briefly present the flip and flex ambiguity with examples.
 
\subsection{Flip and flex ambiguity} \label{Subsec:motivation}
\begin{figure}[]
\centering
\subfigure[]{ \label{ex01_a}
\includegraphics[height = 3.8cm]{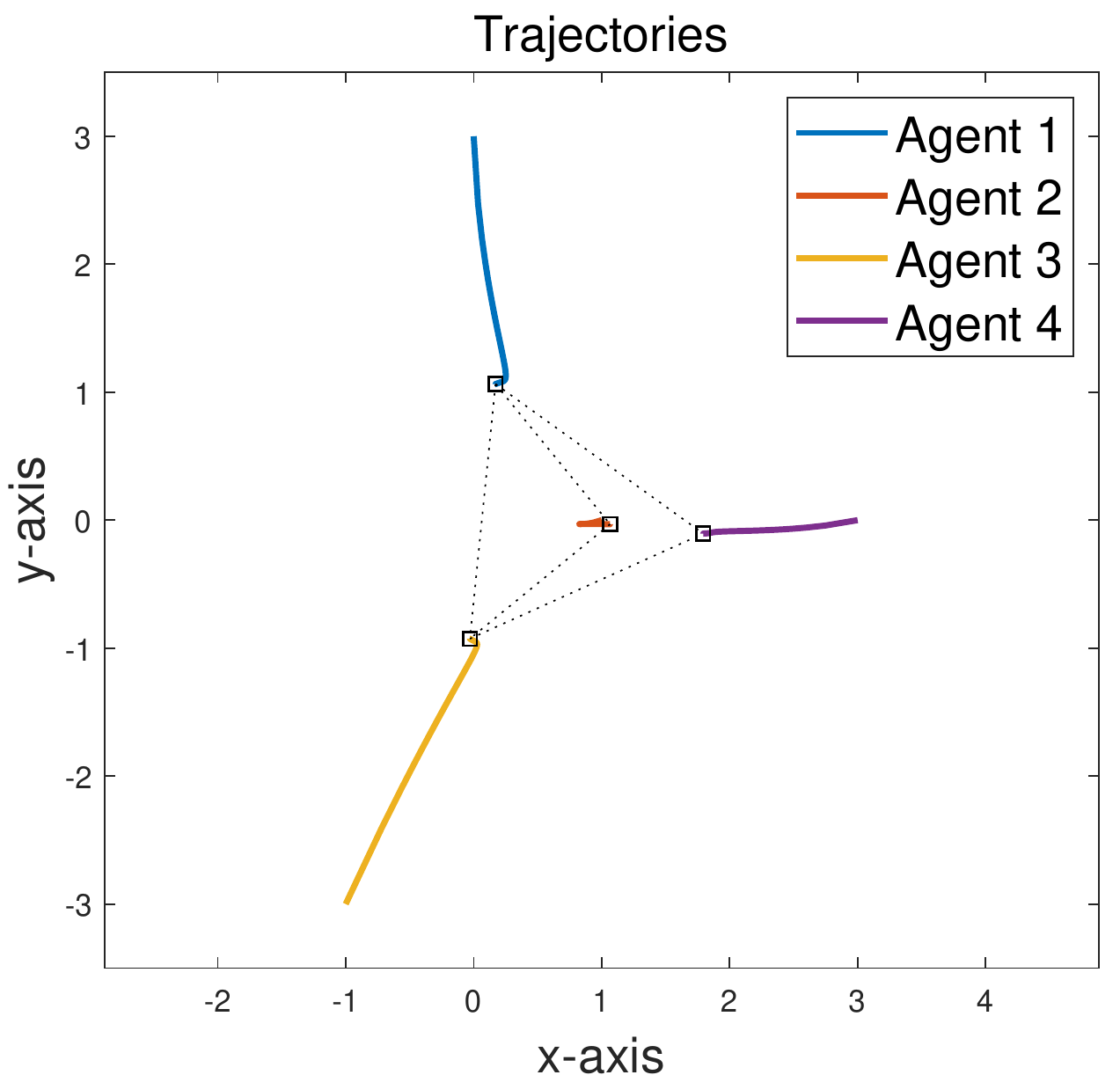}  
}~
\subfigure[]{ \label{ex01_b}
\includegraphics[height = 3.8cm]{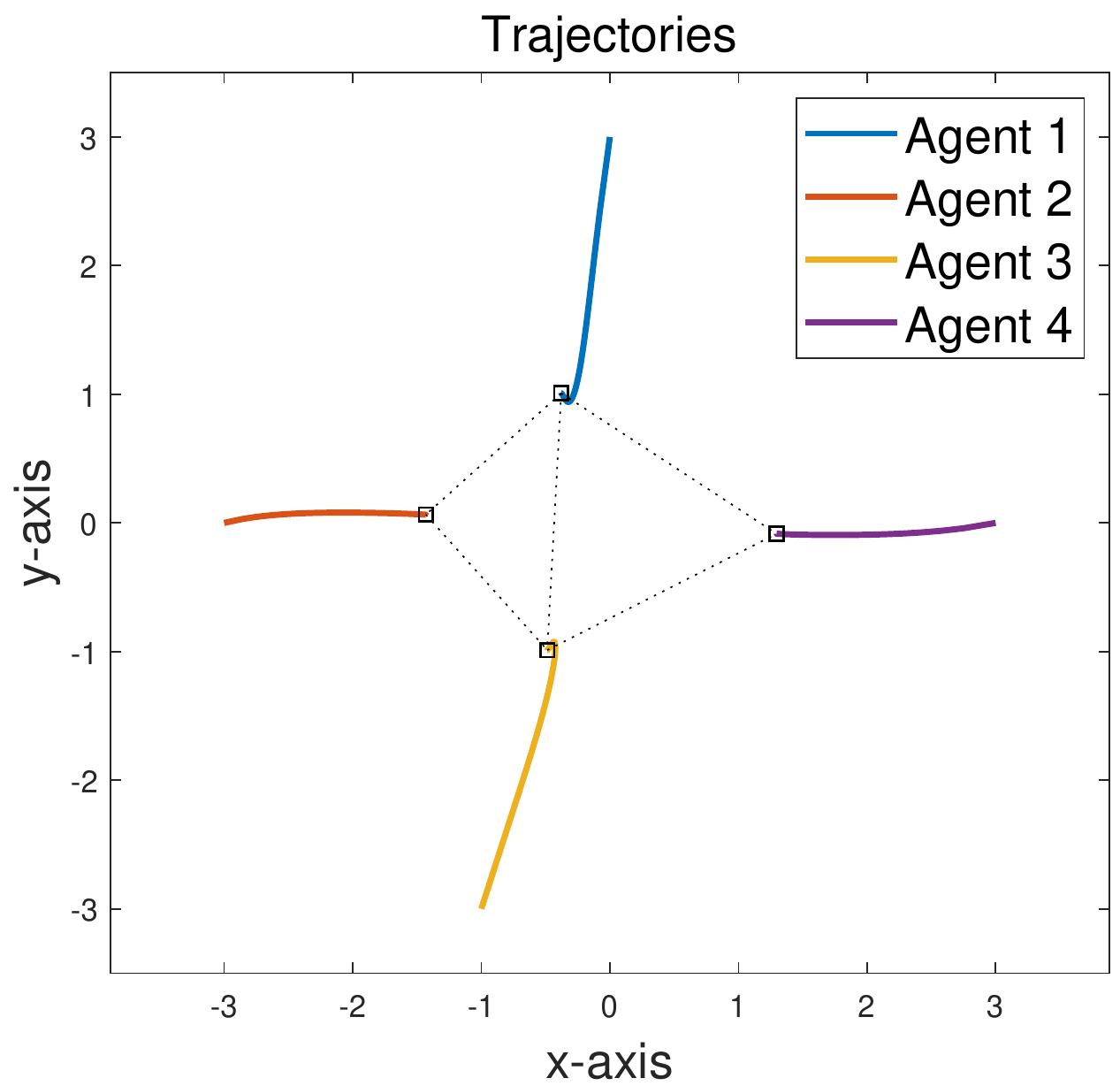}
}
\caption{Simulations on the flip ambiguity under the controller \eqref{control_law_origin01}, where the initial formations for \myfig\ref{ex01_a} and \myfig\ref{ex01_b} are different but the desired distance constraints of both formations are the same. The symbol  $\Box$ denotes the final position of each agent.} \label{example_flip_ambiguity}
\end{figure}

\begin{figure}[]
\centering
\subfigure[]{ \label{ex01_c}
\includegraphics[height = 3.8cm]{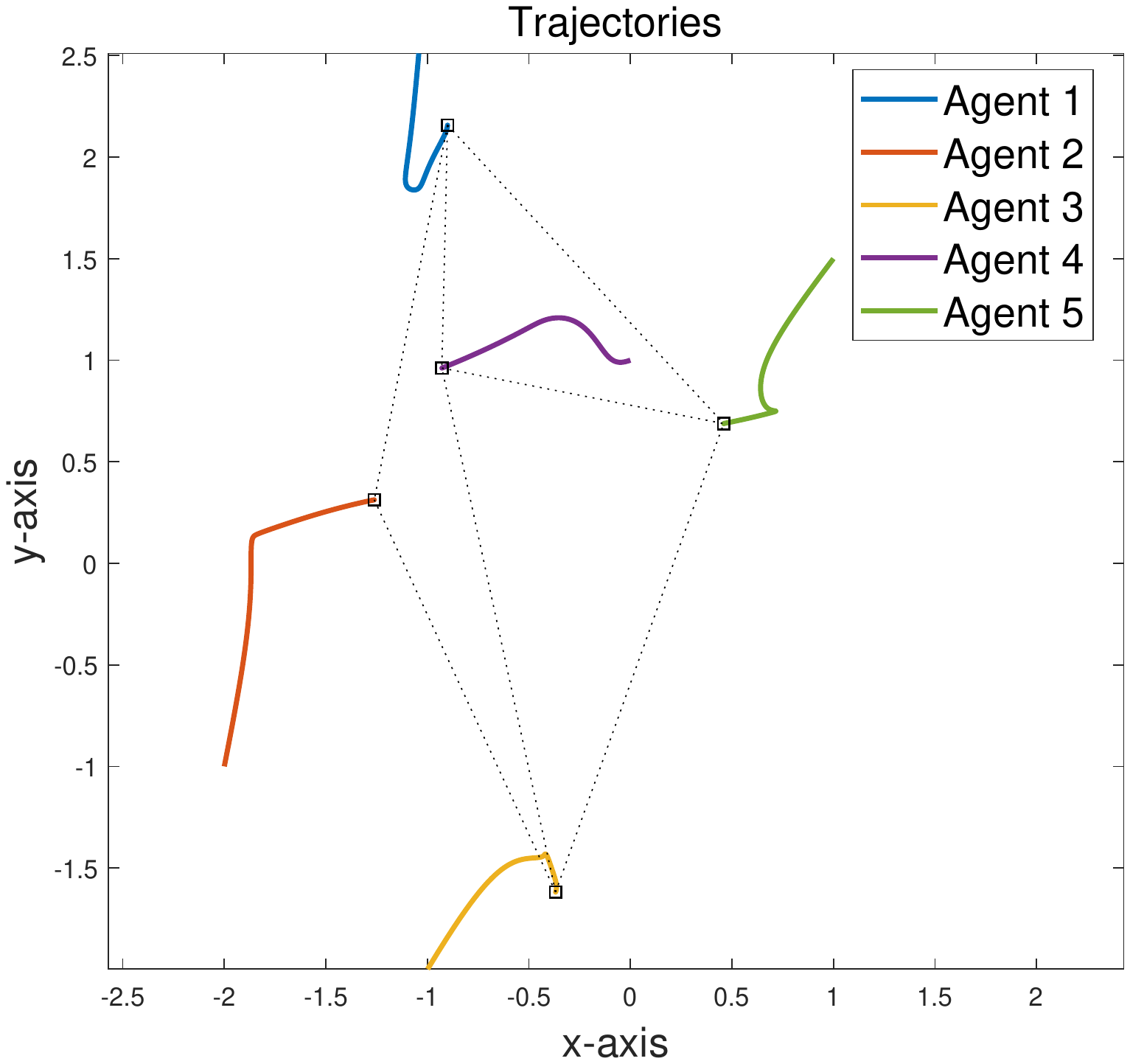}
}~
\subfigure[]{ \label{ex01_d}
\includegraphics[height = 3.8cm]{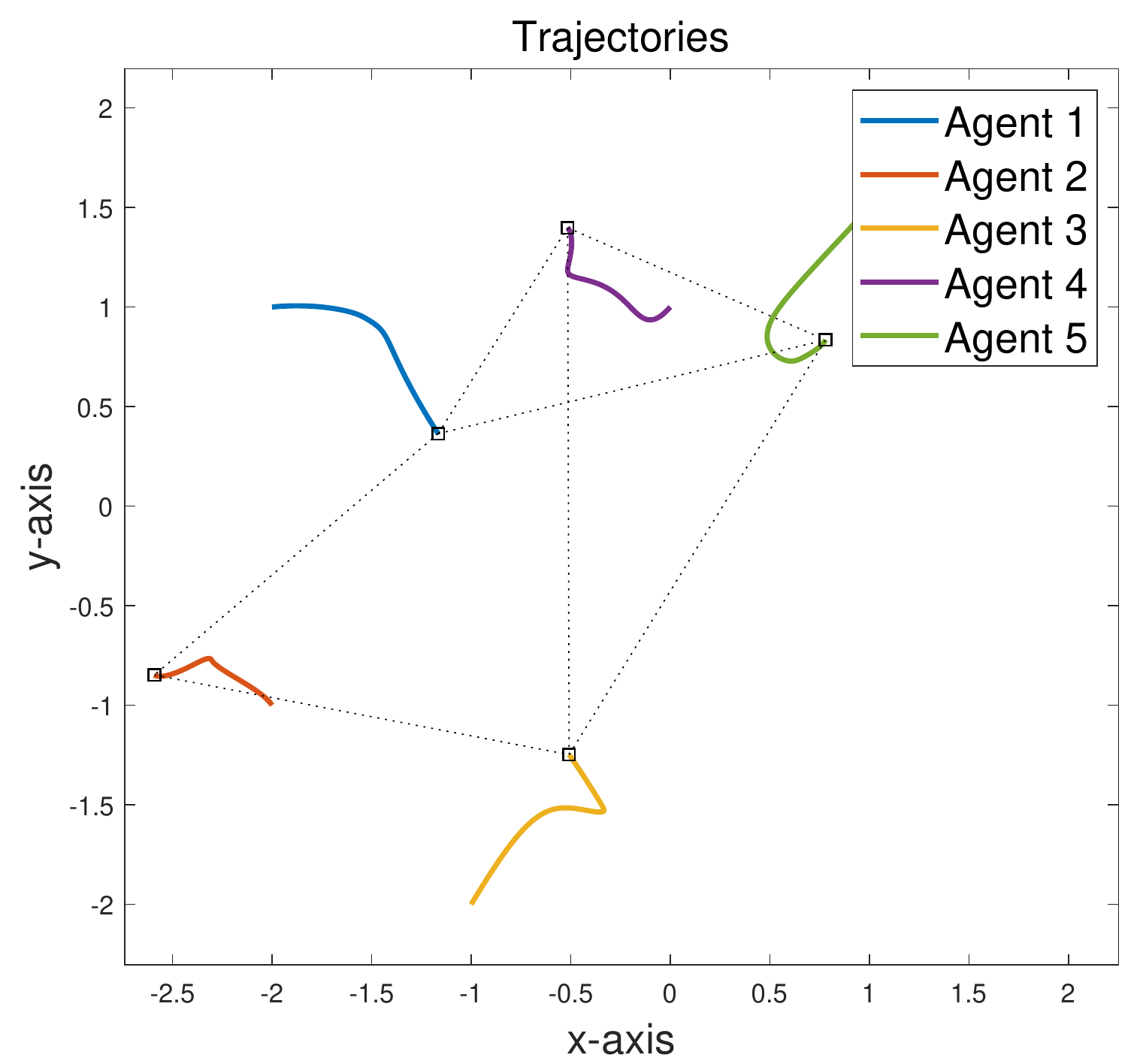}
}
\caption{Simulations on the flex ambiguity under the controller \eqref{control_law_origin01}, where the initial formations for \myfig\ref{ex01_c} and \myfig\ref{ex01_d} are different but the desired distance constraints of both formations are the same. The symbol  $\Box$ denotes the final position of each agent.} \label{example_flex_ambiguity}
\end{figure}
Let us consider a formation control system based on the infinitesimal (distance) rigidity under the control law \eqref{control_law_origin01} as shown in \myfig\ref{example_flip_ambiguity}.
The final formation in \myfig\ref{ex01_a} can be transformed to the final formation in \myfig\ref{ex01_b} by a reflection while the two formations are infinitesimally (distance) rigid and all distance constraints are satisfied. In other words, agent 2 in \myfig\ref{ex01_a} can be flipped over edge (1,3) while the two formations share the same distance constraints.  This kind of ambiguity is called the flip ambiguity.
Note that we also call a reflection of an entire framework in $\mathbb{R}^{2}$ a flip ambiguity since such reflection still causes an another issue called \myemph{ordering issue}.

As a matter of fact, the ordering issue cannot be avoided based on the traditional (distance) rigidity and weak rigidity theories even if each agent can measure relative information of all agents in a system. Let us again consider the simulations in \myfig\ref{example_flip_ambiguity}. Based on the traditional (distance) rigidity theory, we need to add one more distance constraint to the edge (2, 4) to eliminate the flip issue. However, this approach also causes another issue, i.e., the ordering issue of agents; for example, the counterclockwise order of agents in \myfig\ref{Fig:ex_solution_02} is 1, 4, 3 and 2 while it is 1, 2, 3 and 4 in \myfig\ref{Fig:ex_solution_01}.

Next, we consider formation flex ambiguity under distance constraints. An example of flex ambiguity is shown in  \myfig\ref{example_flex_ambiguity}, where the desired distance constraints in \myfig\ref{ex01_c} and \myfig\ref{ex01_d} are the same and the initial positions of agents in both figures are also the same except that of agent 1.
However, the final formations in \myfig\ref{ex01_c} and \myfig\ref{ex01_d} are different while the final formation shape of the triangles composed of agents 3, 4 and 5 in both figures is the same.
In other words, if the distance constraint between agents 1 and 4 in \myfig\ref{ex01_c} is removed then the formation in \myfig\ref{ex01_c} can flex. Also,   agent 1 can move on a circle centered at   agent 5 while other constraints are satisfied. Then, the formation in \myfig\ref{ex01_c} can be the same as the formation in \myfig\ref{ex01_d}. This ambiguity is called the flex ambiguity.

The flip and flex ambiguity shows that we may not achieve the same target formation even if an infinitesimally (distance) rigid or infinitesimally weakly rigid formation satisfies all of the desired constraints.
Such an ambiguity usually occurs when we try to achieve a unique formation shape  with the minimal number of constraints; however, not every infinitesimally (distance) rigid or infinitesimally weakly rigid formation with the minimal number of constraints causes the flip and flex ambiguity. 
\subsection{{\color{black}Idea} on how to avoid the flip and flex ambiguity} \label{Subsec:motivation02}
In this subsection, 
we briefly discuss how we avoid the ambiguities of formations, where the key idea is to include signed constraints such as signed angle constraints   and signed volume constraints.
We suggest to add signed constraints to deal with the flip and flex issue. For example, in \myfig\ref{Fig:ex_solution_03} the formation shape will be unique and does not cause the flip ambiguity including the ordering issue by adding two extra constraints involving signed angles. Similarly, signed constraints can also deal with the flex ambiguity shown in \myfig\ref{example_flex_ambiguity}. We will present a detailed analysis on avoiding flip/flex ambiguity in $\mathbb{R}^{2}$ by signed angle constraints in Section \ref{Sec:inf_rigid} and will propose a new hybrid distance-angle rigidity theory. 

In the case of the 3-dimensional space, we are not sure whether there only exits flip and flex ambiguity in characterizing formations. Thus, in $\mathbb{R}^{3}$, we only address the flip and flex ambiguity by using signed-volume constraints, which will be introduced in Section \ref{Sec:inf_rigid_3D}.
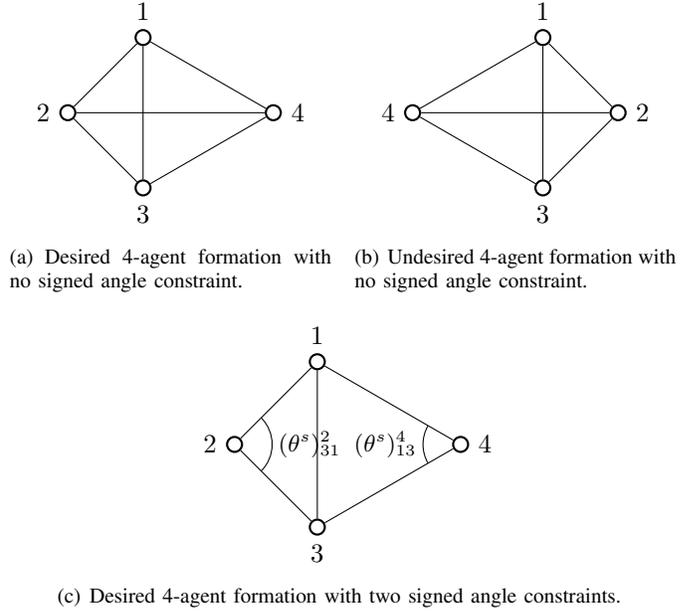
\begin{figure}[]
\centering
\subfigure[Desired 4-agent formation with no signed angle constraint.]{ \label{Fig:ex_solution_01}
\,\,\begin{tikzpicture}[scale=1.0]
\node[place] (node1) at (-1,0) [label=left:$2$] {};
\node[place] (node2) at (0,-1) [label=below:$3$] {};
\node[place] (node3) at (1.732,0) [label=right:$4$] {};
\node[place] (node4) at (0,1) [label=above:$1$] {};

\draw[lineUD] (node1)  --  (node2);
\draw[lineUD] (node2)  -- (node3);
\draw[lineUD] (node2)  -- (node4);
\draw[lineUD] (node1)  -- (node4);
\draw[lineUD] (node3)  -- (node4);
\draw[lineUD] (node1)  -- (node3);
\end{tikzpicture}\,\,
}\,\,~
\subfigure[Undesired 4-agent formation with no signed angle constraint.]{ \label{Fig:ex_solution_02}
\,\,\begin{tikzpicture}[scale=1.0]
\node[place] (node1) at (1,0) [label=right:$2$] {};
\node[place] (node2) at (0,-1) [label=below:$3$] {};
\node[place] (node3) at (-1.732,0) [label=left:$4$] {};
\node[place] (node4) at (0,1) [label=above:$1$] {};

\draw[lineUD] (node1)  --  (node2);
\draw[lineUD] (node2)  -- (node3);
\draw[lineUD] (node2)  -- (node4);
\draw[lineUD] (node1)  -- (node4);
\draw[lineUD] (node3)  -- (node4);
\draw[lineUD] (node1)  -- (node3);
\end{tikzpicture}\,\,
}
\subfigure[Desired 4-agent formation with two signed angle constraints. 
]{ \label{Fig:ex_solution_03}
\qquad\qquad\qquad\begin{tikzpicture}[scale=1.10]
\node[place] (node1) at (-1,0) [label=left:$2$] {};
\node[place] (node2) at (0,-1) [label=below:$3$] {};
\node[place] (node3) at (1.732,0) [label=right:$4$] {};
\node[place] (node4) at (0,1) [label=above:$1$] {};

\draw[lineUD] (node1)  --  (node2);
\draw[lineUD] (node2)  -- (node3);
\draw[lineUD] (node2)  -- (node4);
\draw[lineUD] (node1)  -- (node4);
\draw[lineUD] (node3)  -- (node4);

\pic [draw, -, "\small$(\theta^s)^2_{31}$", angle eccentricity=2.0] {angle = node2--node1--node4};
\pic [draw, -, "\small$(\theta^s)^4_{13}$", angle eccentricity=2.0] {angle = node4--node3--node2};
\end{tikzpicture}\qquad\qquad\qquad
}
\caption{Example of formation constraints  to avoid the flip ambiguity in $\mathbb{R}^2$ where the solid edges denote distance constraints, and $(\theta^s)^i_{jk}$ denotes a signed angle constraint, i.e., $(\theta^s)^i_{jk}=\angle ijk$ from edge $(i,j)$ to edge $(i,k)$ and $(\theta^s)^i_{jk}\neq (\theta^s)^i_{kj}$.}\label{Fig:ex_solution}
\end{figure} 

\section{Hybrid distance-angle rigidity theory in $\mathbb{R}^{2}$: infinitesimal rigidity with distance (or angle) and signed-angle constraints} \label{Sec:inf_rigid}
In this section, we introduce two modified concepts from the infinitesimal (distance) rigidity and  infinitesimal weak rigidity in $\mathbb{R}^{2}$. One is an infinitesimal rigidity theory with distance and signed constraints; the other one is a concept of infinitesimal rigidity with angle and signed constraints. These two new rigidity theories are related to defining and achieving a unique formation up to rigid-body translations and rotations (and additionally scaling in the case of infinitesimal rigidity with angle and signed constraints) of a formation, with the aim to resolve flip and flex formation ambiguities. 

Firstly, we define a signed constraint defined by bearing vectors as follows
\begin{align}  \label{eq:main_signed_area}
\bar{S}_{ijk}= \text{det} \begin{bmatrix} \frac{p_j-p_i}{\norm{p_j-p_i}} & \frac{p_k-p_i}{\norm{p_k-p_i}} \end{bmatrix}, \quad (i,j),(i,k) \in \mathcal{E}.
\end{align}
The concept of signed constraints is motivated by \cite{anderson2017formation} even though the authors in \cite{anderson2017formation}
utilized a signed area constraint widely known as 
$
S_{ijk}= \text{det} \begin{bmatrix} p_j-p_i & p_k-p_i \end{bmatrix}, (i,j),(i,k) \in \mathcal{E}.
$
The signed constraint $\bar{S}_{ijk}$ defined in \eqref{eq:main_signed_area} is equivalently expressed as
\begin{align} \label{eq_signed_02}
\bar{S}_{ijk}&= \text{det} \begin{bmatrix} \frac{p_j-p_i}{\norm{p_j-p_i}} & \frac{p_k-p_i}{\norm{p_k-p_i}} \end{bmatrix} \nonumber \\
&=\left\|\frac{p_j-p_i}{\norm{p_j-p_i}}\right\| \left\|\frac{p_k-p_i}{\norm{p_k-p_i}}\right\|  \sin((\theta^s)^i_{jk}) \nonumber\\ 
&= \sin((\theta^s)^i_{jk}) , \quad (i,j),(i,k) \in \mathcal{E},
\end{align}
where $(\theta^s)_{jk}^{i} \in [0,2\pi)$ denotes a signed angle.
We call the constraint in \eqref{eq:main_signed_area} a \myemph{signed angle constraint} due to the fact that  $\sin(\theta^i_{jk})=- \sin(\pi+\theta^i_{jk}), \forall \theta_{jk}^{i} \in [0,\pi]$.
\begin{remark}
The rigidity concepts studied in \cite{jing2018weak,kwon2018generalized} only include unsigned angle constraints. Specifically, the authors in \cite{jing2018weak} use inner products; the rigidity concept in  \cite{kwon2018generalized} involves cosine functions as the unsigned angle constraints. Thus, signed angle constraints differ from unsigned angle constraints.
\end{remark}

Moreover, with  the \myemph{perpendicular matrix} $J$ defined as $J \triangleq \begin{bmatrix} 0 & 1 \\ -1 & 0\end{bmatrix}$, the determinant function $\text{det} \begin{bmatrix} \frac{p_j-p_i}{\norm{p_j-p_i}} & \frac{p_k-p_i}{\norm{p_k-p_i}} \end{bmatrix}$ can be represented by $\left(\frac{{z}_{ji}^\top}{\norm{z_{ji}}}J\frac{z_{ki}}{\norm{z_{ki}}}\right)$.
Thus,  it holds that $\left(\frac{{z}_{ji}^\top}{\norm{z_{ji}}}J\frac{z_{ki}}{\norm{z_{ki}}}\right)=\sin((\theta^s)^i_{jk})$ for all pairs of incident  edges  $(i,j),(i,k) \in \mathcal{E}$.
Then, we define a set $\mathcal{S}$ for signed angle constraints as 
$\mathcal{S} = \set{(i,j,k)  \in \mathcal{V}^3 \given i,j,k \in \mathcal{V}\text{  to characterize } \bar{S}_{ijk}}$
with $m_s=\card{\mathcal{S}}$.
 Moreover, we occasionally use an ordered set $\bar{S}_{ijk}$ such that $\bar{S}_w \triangleq \bar{S}_{ijk},\forall w\in \{ 1,..., m_s\}$. 
The signed angle constraint is also used as a constraint in a framework together with distance or angle constraints, and is denoted by $(\theta^s)^i_{jk}, (i,j,k) \in \mathcal{S}$.
 For a clear understanding, we provide  an example by using a framework with $\mathcal{S}$ in \myfig\ref{Fig:ex_set}. Then, we are ready to state the main concepts of this paper in the next two subsections.

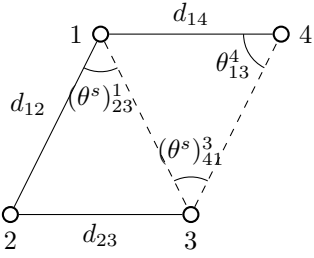
\begin{figure}[]
\centering
\begin{tikzpicture}[scale=1.2]
\node[place] (node1) at (0,2) [label=left:$1$] {};
\node[place] (node2) at (-1,0) [label=below:$2$] {};
\node[place] (node3) at (1,0) [label=below:$3$] {};
\node[place] (node4) at (2,2) [label=right:$4$] {};

\draw[lineUD] (node1)  -- node [above left] {$d_{12}$} (node2);
\draw[lineUD] (node2)  -- node [below] {$d_{23}$} (node3);
\draw[dashed] (node1)  -- (node3);
\draw[lineUD] (node1)  -- node [above] {$d_{14}$} (node4);
\draw[dashed] (node3)  -- (node4);

\pic [draw, -, "$\theta^4_{13}$", angle eccentricity=1.5] {angle = node1--node4--node3};
\pic [draw, -, "$(\theta^s)^1_{23}$", angle eccentricity=1.7] {angle = node2--node1--node3};
\pic [draw, -, "$(\theta^s)^3_{41}$", angle eccentricity=1.7] {angle = node4--node3--node1};
\end{tikzpicture}
\caption{Example of a framework $(\mathcal{G},p)$ with $\mathcal{S}$: 4-agent formation with $\mathcal{E}=\set{(1,2),(1,4),(2,3)}$, $\mathcal{A}=\set{(4,1,3)}$ and $\mathcal{S}=\set{(1,2,3),(3,4,1)}$, which means that the formation has
3 distance, 1 angle and 2 signed angle constraints. The dashed lines denote virtual edges; that is, they are not distance constraints.}\label{Fig:ex_set}
\end{figure} 
\subsection{Infinitesimal rigidity with distance and signed constraints} \label{Inf_rigid_DS}
In this subsection, we consider a framework with hybrid distance and signed angle constraints to develop a new rigidity theory. 
We first define the following \myemph{distance-sign rigidity function} $F^s_D: \chi \rightarrow \mathbb{R}^{(m_d+m_s)}$ for some properly defined $\chi\subset\mathbb{R}^{2n}$:
\begin{align}\label{Eq:sign_rigidity_fn}
F^s_{D}(p) \triangleq& \left[\frac{1}{2} \norm{z_{1}}^2, ... ,\frac{1}{2}\norm{z_{m_d}}^2, \bar{S}_1, ... ,\bar{S}_{m_s} \right]^\top.
\end{align}
Let us consider the following distance and signed angle constraints.
\begin{align} 
\frac{1}{2}\norm{p_{i}-p_{j}}^2&=constant, \quad\forall (i,j) \in \mathcal{E},  \label{eq:inf_const01}\\
\bar{S}_{ijk}&= constant, \quad \forall (i,j,k) \in \mathcal{S}, \label{eq:inf_const02}
\end{align}
where the constraints should be consistent and physically realizable. Then, the time derivative of \eqref{eq:inf_const01} is given by
\begin{align} 
\left(p_i-p_j \right)^\top \left(v_i-v_j \right) =0, \quad\forall (i,j) \in \mathcal{E}, \label{eq:inf_motions01}
\end{align}
and, with utilizing the perpendicular matrix $J$, the time derivative of \eqref{eq:inf_const02}  is given as
\begin{align} 
&\frac{(z_{ji})^\top }{\norm{z_{ji}}}J\frac{P_{z_{ki}}}{\norm{z_{ki}}} (v_k-v_i)
+\frac{(z_{ki})^\top }{\norm{z_{ki}}} J^\top \frac{P_{z_{ji}}}{\norm{z_{ji}}} (v_j-v_i)\nonumber \\
&=0, \quad \forall (i,j,k) \in \mathcal{S}, \label{eq:inf_motions02}
\end{align}
where the symbol $v_i$, $v_j$ and $v_k$ denote infinitesimal motions of $p_i$, $p_j$ and $p_k$, respectively.
Moreover, the equations \eqref{eq:inf_motions01} and \eqref{eq:inf_motions02} can be written by a matrix form as follows
\begin{align} \label{eq:matrix_infini}
\dot{F}^s_D=\frac{\partial F^s_{D}(p)}{\partial p}\dot{p}=\begin{bmatrix}
     \frac{\partial \mbf{D}}{\partial p}\\ \\
	\frac{\partial \mbf{S}}{\partial p}
  \end{bmatrix}\dot{p}=0,
\end{align}  
where $\mbf{D} = \left[\frac{1}{2}\norm{z_{1}}^2, ... ,\frac{1}{2}\norm{z_{m_d}}^2\right]^\top \in \mathbb{R}^{m_d}$ and $\mbf{S} = \left[\bar{S}_1,...,\bar{S}_{m_s}\right]^\top \in \mathbb{R}^{m_s}$.
We here define a \myemph{distance-sign rigidity matrix} as the Jacobian of the distance-sign rigidity function as below
\begin{equation}\label{dist_sign_rigidity_matrix}
R^s_{D}(p) \triangleq \frac{\partial F^s_{D}(p)}{\partial p} = 
\begin{bmatrix}
     \frac{\partial \mbf{D}}{\partial p}\\ \\
	\frac{\partial \mbf{S}}{\partial p}
  \end{bmatrix}
  =\begin{bmatrix}
     R_D\\ \\
	R_S
  \end{bmatrix} \in 
\mathbb{R}^{(m_d+m_s) \times 2n}.
\end{equation}

We now state a fundamental concept with the result \eqref{eq:matrix_infini}.
If infinitesimal motions of an entire framework, i.e. $\dot{p}$, correspond to only rigid-body translations and rigid-body rotations of the given framework, then we call the motions \myemph{trivial}. Furthermore, we define a concept of \myemph{infinitesimal rigidity with distance and signed constraints (IRDS)} if all infinitesimal motions of the framework are trivial.
To characterize  the concept of the IRDS in an algebraic manner, we prove the following lemma and theorem.
\begin{lemma}\label{Lem:null_rank condition_dist}
It holds that $\vspan\set{\mathds{1}_n\otimes I_2,(I_n\otimes J)p} \subseteq \vnull(R^s_D(p))$ and $\rank\left(R^s_D(p)\right) \leq 2n-3$ for a framework $(\mathcal{G},p)$ with $\mathcal{S}$ in $\mathbb{R}^2$. 
\end{lemma}
\begin{proof}
First, we define an edge set as
\begin{align}
\mathcal{E}'=\mathcal{E} \cup \mathcal{E}_s
\end{align}
with $m_t=\card{\mathcal{E}'}$  where $\mathcal{E}_s=\set{(i,j),(i,k) \given (i,j,k) \in \mathcal{S}}$.
Also, we define an induced graph $\mathcal{G}'=(\mathcal{V},\mathcal{E}')$. Moreover, the newly ordered relative position is defined as $z'_s \triangleq z_{ij}$ for $(i,j) \in \mathcal{E}'$ and $s \in \{ 1,..., m_t\}$.
Then, we remark that $R_S=\frac{\partial \mbf{S}}{\partial p}=\frac{\partial \mbf{S}}{\partial z'} \frac{\partial z'}{\partial p}=\frac{\partial \mbf{S}}{\partial z'}\bar{H}'$ where $\bar{H}'=(H'\otimes I_2)$ and $z' = \big[{z'}_{1}^\top,...,{z'}_{m_t}^\top \big]^\top = (H'\otimes I_2)p \in \mathbb{R}^{2m_d}$. 
Since $H'$ is an incidence matrix, it is obvious that $\vspan\set{\mathds{1}_n\otimes \begin{bmatrix} 1 & 0  \end{bmatrix}^\top} \subseteq \vnull(H'\otimes I_2)$ and $\vspan\set{\mathds{1}_n\otimes \begin{bmatrix} 0 & 1  \end{bmatrix}^\top} \subseteq \vnull(H'\otimes I_2)$ \cite{mesbahi2010graph}. Thus, it holds that $\vspan\set{\mathds{1}_n\otimes I_2} \subseteq \vnull(R_S(p))$.

Next, let us consider a $w$-th element in $\mbf{S}$, i.e., $\bar{S}_w, w\in \{ 1,..., m_s\}$. 
Then, we check the following equation with the fact that there exist $i$ and $j$ such that $\bar{S}_w= \left(\frac{{z'}_i^\top}{\norm{z'_i}}J\frac{z'_j}{\norm{z'_j}}\right)$ and $i,j \in \{ 1,..., m_t\}$:
\begin{flalign}
\frac{\partial \bar{S}_w}{\partial p} (I_n\otimes J)p
&=\frac{\partial \bar{S}_w}{\partial z'} \frac{\partial z'}{\partial p} (I_n\otimes J)p \nonumber \\
&=\frac{\partial \bar{S}_w}{\partial z'} \bar{H'}  (I_n\otimes J)p \nonumber \\
&=\frac{\partial}{\partial z'}\left( \frac{{z'}_i^\top}{\norm{z'_i}}J\frac{z'_j}{\norm{z'_j}}\right) \bar{H}'  (I_n\otimes J)p. \label{temp_eq01}
\end{flalign} 
We remark that $\bar{H}'(I_n\otimes J)p=\left[(J z'_1)^\top, \cdots, (J z'_{m_t})^\top \right]^\top$ since the following holds true
\begin{flalign}\label{eq:ro_z}
\bar{H'}(I_n\otimes J)p &= (H'\otimes I_3)(I_n\otimes J)p  =(H'\otimes J)p \nonumber \\
&= (I_{m_t} H'\otimes J I_3)p \nonumber \\
&= (I_{m_t} \otimes J)(H'\otimes I_3)p \nonumber \\
&=(I_{m_t} \otimes J)z'= \begin{bmatrix}
J z'_1 \\
\vdots\\
J z'_{m_t}
\end{bmatrix}.
\end{flalign}
Therefore, the equation \eqref{temp_eq01} can be calculated as
\begin{flalign}
&\frac{\partial}{\partial z'}\left( \frac{{z'}_i^\top}{\norm{z'_i}}J\frac{z'_j}{\norm{z'_j}}\right)\begin{bmatrix}
J z'_1 \\
\vdots\\
J z'_{m_t}
\end{bmatrix} & \nonumber \\
&=\frac{\partial}{\partial z'_i} \left(\frac{{z'}_i^\top}{\norm{z'_i}}J\frac{z'_j}{\norm{z'_j}}\right) J z'_i +  \frac{\partial}{\partial z'_j}\left(\frac{{z'}_i^\top}{\norm{z'_i}}J\frac{z'_j}{\norm{z'_j}}\right)J z'_j & \nonumber \\
&=\frac{1}{{\norm{z'_i}}\norm{z'_j}}\left({z'}_j^\top z'_i - {z'}_i^\top {z'}_j \right) = 0, &
\end{flalign} 
where we have used the result of Lemma \ref{Lem:Jaco_sine} in Appendix and the fact that ${z'}_i^\top J z'_i={z'}_j^\top J z'_j=0$.
Thus, for all $w\in \{ 1,..., m_s\}$, it holds that $\frac{\partial \mbf{S}}{\partial p} (I_n\otimes J)p=R_S(I_n\otimes J)p =0$ and this implies that $\vspan\set{(I_n\otimes J)p }  \subseteq \vnull(R_S(p))$. 

From \cite[Lemma 1]{sun2017distributed}, it is always true that $\vspan\set{\mathds{1}_n\otimes I_2,(I_n\otimes J)p} \subseteq \vnull(R_D(p))$. Consequently, it always holds that  $\vspan\set{\mathds{1}_n\otimes I_2,(I_n\otimes J)p} \subseteq\vnull\left( \begin{bmatrix} R_D\\R_S  \end{bmatrix}\right) = \vnull(R^s_D(p))$, which implies that $\rank\left(R^s_D(p)\right) \leq 2n-3$.
\end{proof}

Note that  $\vspan\set{\mathds{1}_n\otimes I_2}$ and $\vspan\set{(I_n\otimes J)p}$ represent rigid-body translations and rotations for an entire framework $(\mathcal{G},p)$, respectively. Then, we reach the following one of two main results in this section.

\begin{theorem}
\label{Thm:Inf_Rank_dist} 
A framework $(\mathcal{G},p)$ with $\mathcal{S}$ is IRDS in $\mathbb{R}^{2}$ if and only if  $\rank(R^s_D(p))=2n-3$. 
\end{theorem}
\begin{proof}
It is true from Lemma \ref{Lem:null_rank condition_dist} that $\vspan\set{\mathds{1}_n\otimes I_2,(I_n\otimes J)p} \subseteq \vnull(R^s_D(p))$ and $\rank\left(R^s_D(p)\right) \leq 2n-3$. Moreover, $\mathds{1}_n\otimes I_2$ and $(I_n\otimes J)p$ denote a rigid-body translation and a rigid-body rotation of the framework. Together with Lemma \ref{Lem:null_rank condition_dist}, these facts imply that $\rank(R^s_D(p))=2n-3$ if and only if all infinitesimal motions satisfying \eqref{eq:matrix_infini} are trivial. 
\end{proof}

Example: we provide one example as shown in \myfig\ref{Fig:ex_rank01} to illustrate the application of Theorem \ref{Thm:Inf_Rank_dist} in determining IRDS for a given framework. The formations in \myfig\ref{Fig:rank_01} and \myfig\ref{Fig:rank_02} satisfy $\rank(R^s_D(p))=5$ and $\rank(R^s_D(p))=4$, respectively, at the same position, which means the framework in \myfig\ref{Fig:rank_01} is IRDS and the framework in \myfig\ref{Fig:rank_02} is not IRDS. In fact, consistent with the constraints as shown in \myfig\ref{Fig:rank_02}, agent 4 can move on an arc of a circle passing through 1,3, and 4.

\begin{figure}[]
\centering
\subfigure[$\rank(R^s_D(p))=5$]{ \label{Fig:rank_01}
\,\,\begin{tikzpicture}[scale=1]
\node[place] (node1) at (0,2) [label=left:$1$] {};
\node[place] (node2) at (-1,0) [label=below:$2$] {};
\node[place] (node3) at (1,0) [label=below:$3$] {};
\node[place] (node4) at (2,2) [label=right:$4$] {};

\draw[lineUD] (node1)  -- node [above left] {$d_{12}$} (node2);
\draw[lineUD] (node2)  -- node [below] {$d_{23}$} (node3);
\draw[dashed] (node1)  -- (node3);
\draw[lineUD] (node1)  -- node [above] {$d_{14}$} (node4);
\draw[dashed] (node3)  -- (node4);

\pic [draw, -, "$(\theta^s)^1_{23}$", angle eccentricity=1.7] {angle = node2--node1--node3};
\pic [draw, -, "$(\theta^s)^4_{13}$", angle eccentricity=1.9] {angle = node1--node4--node3};
\end{tikzpicture}\,\,%
}~
\subfigure[$\rank(R^s_D(p))=4$]{ \label{Fig:rank_02}
\,\,\begin{tikzpicture}[scale=1]
\node[place] (node1) at (0,2) [label=left:$1$] {};
\node[place] (node2) at (-1,0) [label=below:$2$] {};
\node[place] (node3) at (1,0) [label=below:$3$] {};
\node[place] (node4) at (2,2) [label=right:$4$] {};

\draw[lineUD] (node1)  -- node [above left] {$d_{12}$} (node2);
\draw[lineUD] (node2)  -- node [below] {$d_{23}$} (node3);
\draw[dashed] (node1)  -- (node3);
\draw[dashed] (node1)  -- (node4);
\draw[dashed] (node3)  -- (node4);

\pic [draw, -, "$(\theta^s)^1_{23}$", angle eccentricity=1.7] {angle = node2--node1--node3};
\pic [draw, -, "$(\theta^s)^4_{13}$", angle eccentricity=1.9] {angle = node1--node4--node3};
\end{tikzpicture}\,\,%
}
\caption{Example of IRDS and not IRDS formations characterized by different constraints, where the positions are chosen as $p_1=[0\quad 3]^\top, p_2=[-2\quad 0]^\top, p_3=[2\quad 0]^\top$ and $p_4=[4\quad 3]^\top$.}\label{Fig:ex_rank01}
\end{figure}
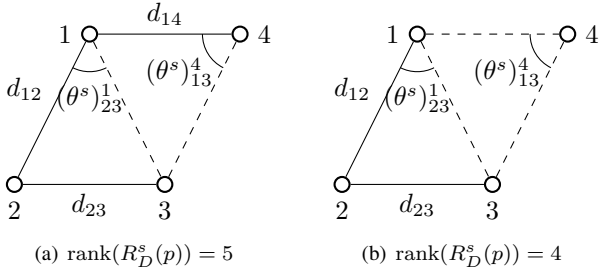 
\subsection{Infinitesimal rigidity with angle and signed constraints}  \label{Inf_rigid_AS}
This subsection provides a similar result to that discussed in Subsection \ref{Inf_rigid_DS} using hybrid angle and signed angle constraints.
First, the \myemph{angle-sign rigidity function} $F^s_A: \chi \rightarrow \mathbb{R}^{(m_a+m_s)}$  for some properly defined $\chi\subset\mathbb{R}^{2n}$
is defined as follows
\begin{align}
F^s_{A}(p) \triangleq& \left[ A_1, ... ,A_{m_a}, \bar{S}_1, ... ,\bar{S}_{m_s} \right]^\top.
\end{align}
We next consider the following time derivative of the constraint $\cos{\theta_{jk}^{i}}=\frac{z^\top_{ji}}{\norm{z_{ji}}}\frac{z_{ki}}{\norm{z_{ki}}}= constant, \, \forall (i,j,k) \in \mathcal{A}$:
\begin{align} 
&\frac{z_{ki}^\top}{\norm{z_{ki}}} \frac{P_{z_{ji}}}{\norm{z_{ji}}} (v_j-v_i)
+ \frac{z^\top_{ji}}{\norm{z_{ji}}}  \frac{P_{z_{ki}}}{\norm{z_{ki}}} \left(v_k-v_i \right) = 0,  \nonumber \\
& \forall (i,j,k) \in \mathcal{A} \label{eq:inf_motions03}
\end{align}
where $v_i$ is also an infinitesimal motion of $p_i$. Note that the cosine function is used as an unsigned constraint while the sine function \eqref{eq_signed_02} is considered as a signed constraint. The time derivative of the signed angle constraint is the same as \eqref{eq:inf_motions02}.
Then, the equations \eqref{eq:inf_motions02} and \eqref{eq:inf_motions03} can also be written by a matrix form as follows
\begin{align}\label{eq:matrix_infini_ang}
\dot{F}^s_A=\frac{\partial F^s_{A}(p)}{\partial p}\dot{p}=\begin{bmatrix}
     \frac{\partial \mbf{A}}{\partial p}\\ \\
	\frac{\partial \mbf{S}}{\partial p}
  \end{bmatrix}\dot{p}=\begin{bmatrix}
     R_W\\ \\
	R_S
  \end{bmatrix}\dot{p}=R^s_{A}(p)\dot{p}=0,
\end{align} 
where $\mbf{A} = \left[A_1, ... ,A_{m_a}\right]^\top \in \mathbb{R}^{m_a}$ and $R^s_{A}(p)\in 
\mathbb{R}^{(m_a+m_s) \times 2n}$ denotes the \myemph{angle-sign rigidity matrix}.

From the same viewpoint as Subsection \ref{Inf_rigid_DS}, if infinitesimal motions of an entire framework, $\dot{p}$, correspond to only rigid-body translations, rotations and additionally scaling motions of the framework, then we also call the motions \myemph{trivial}. Moreover, we can also define a concept of \myemph{infinitesimal rigidity with angle and signed constraints (IRAS)} if all infinitesimal motions of the framework are trivial. In the same spirit as in  Subsection \ref{Inf_rigid_DS}, we prove the following lemma and theorem.

\begin{lemma}\label{Lem:null_rank condition_ang}
It holds that $\vspan\set{\mathds{1}_n\otimes I_2,(I_n\otimes J)p,p} \subseteq \vnull(R^s_A(p))$ and $\rank\left(R^s_A(p)\right) \leq 2n-4$ for a framework $(\mathcal{G},p)$ with $\mathcal{S}$ in $\mathbb{R}^2$. 
\end{lemma}
\begin{proof}
This proof is similar to that of Lemma \ref{Lem:null_rank condition_dist}.
We firstly define a new edge set $\mathcal{E}'$ as 
\begin{align}
\mathcal{E}'=\{(i,j),(i,k) \mid (i,j,k) \in \mathcal{A} \lor (i,j,k) \in \mathcal{S}\} 
\end{align}
with $m_t=\card{\mathcal{E}'}$.
Also, we define an induced graph $\mathcal{G}'=(\mathcal{V},\mathcal{E}')$.
Moreover, the newly ordered relative position vector consistent with the incidence matrix is defined as $z'_s=z_{ij}$  for $(i,j) \in \mathcal{E}', \forall s\in \{ 1,..., m_t\}$. Then, we have
 $z' = \big[{z'}_{1}^{\top},...,{z'}_{m_t}^{\top} \big]^\top = (H'\otimes I_2)p \in \mathbb{R}^{2m_t}$.

Let us consider the $w$-th element in $\mbf{S}$, i.e., $\bar{S}_w, w\in \{ 1,..., m_s\}$.
With the fact that there exist $i$ and $j$ such that $\bar{S}_w= \left(\frac{{z'}_i^\top}{\norm{z'_i}}J\frac{z'_j}{\norm{z'_j}}\right)$ and $i,j \in \{ 1,..., m_t\}$,
we can derive
\begin{flalign}
\frac{\partial S_w}{\partial p} p
&=\frac{\partial S_w}{\partial z'} \frac{\partial z'}{\partial p} p
=\frac{\partial S_w}{\partial z'} (H'\otimes I_2) p \nonumber \\
&=\frac{\partial}{\partial z'}\left( \frac{{z'}_i^\top}{\norm{z'_i}}J\frac{z'_j}{\norm{z'_j}}\right)\begin{bmatrix}
z'_1 \\
\vdots\\
z'_{m_t}
\end{bmatrix} &\nonumber \\
&=\frac{\partial}{\partial z'_i} \left(\frac{{z'}_i^\top}{\norm{z'_i}}J\frac{z'_j}{\norm{z'_j}}\right)z'_i + \frac{\partial}{\partial z'_j}\left(\frac{{z'}_i^\top}{\norm{z'_i}}J\frac{z'_j}{\norm{z'_j}}\right) z'_j &\nonumber \\
&=0+0=0,&
\end{flalign} 
where we have used the result of Lemma \ref{Lem:Jaco_sine} in Appendix and it holds that $P_{z'_i} z'_i=P_{z'_j} z'_j=0$.
Thus, for all $w\in \{ 1,..., m_s\}$, one concludes that $R_S(p) p =0$ which implies that $\vspan\set{p}  \subseteq \vnull(R_S(p))$. Moreover,  it also holds that $\vspan\set{\mathds{1}_n\otimes I_2,(I_n\otimes J)p} \subseteq \vnull(R_S(p))$ by the same proof as in Lemma \ref{Lem:null_rank condition_dist}. Therefore, we have $\vspan\set{\mathds{1}_n\otimes I_2,(I_n\otimes J)p,p} \subseteq \vnull(R_S(p))$.

From  \cite[Lemma 3.3]{kwon2018infinitesimal_inp.}, one obtains that $\vspan\set{\mathds{1}_n\otimes I_2,(I_n\otimes J)p,p} \subseteq \vnull(R_W(p))$ when $\mathcal{E} = \emptyset$. Consequently, it always holds that  $\vspan\set{\mathds{1}_n\otimes I_2,(I_n\otimes J)p,p} \subseteq \vnull(R^s_A(p))$, which implies that $\rank\left(R^s_A(p)\right) \leq 2n-4$.
\end{proof}

Note that $\vspan\set{p}$ denotes scaling motions for an entire framework $(\mathcal{G},p)$. Then, we state the second main theorem in this section as follows. 

\begin{theorem}
\label{Thm:Inf_Rank_ang} 
A framework $(\mathcal{G},p)$ with $\mathcal{S}$ is IRAS in $\mathbb{R}^{2}$ if and only if  $\rank(R^s_A(p))=2n-4$. 
\end{theorem}
\begin{proof}
We have that $\vspan\set{\mathds{1}_n\otimes I_2,(I_n\otimes J)p,p} \subseteq \vnull(R^s_A(p))$ and $\rank\left(R^s_A(p)\right) \leq 2n-4$ from Lemma \ref{Lem:null_rank condition_ang}. Also note that $\mathds{1}_n\otimes I_2$, $(I_n\otimes J)p$ and $p$ correspond to  a rigid-body translation, a rigid-body rotation and a scaling of an entire framework, respectively, which means that $\rank(R^s_A(p))=2n-4$ if and only if all infinitesimal motions satisfying \eqref{eq:matrix_infini_ang} are trivial. Therefore, we conclude the statement.
\end{proof}

\subsection{Discussion on minimal rigidity with hybrid and signed constraints in $\mathbb{R}^{2}$} \label{Subsec:number_const_2D}
In this subsection, we discuss how to determine the minimal number of constraints to locally determine a unique formation shape  up to translations and rotations (and additionally scaling factors in the case of  IRAS formations) of an entire formation. Moreover, we observe another ambiguous phenomenon distinct from the flip and flex ambiguity.

We first state the minimal and non-minimal number of constraints for achieving locally unique formations.
In the case of formations with distance and signed angle constraints, if a framework $(\mathcal{G},p)$ satisfies $\rank(R^s_D(p))=2n-3$ in $\mathbb{R}^2$ and $\rank(R^s_D(p))$ is exactly equal to the number of distance and signed constraints, i.e. $\rank(R^s_D)=2n-3=m_d+m_s$, then the number $2n-3$ is the minimal number to have a fixed formation shape (up to translation and rotation).
Similarly, in the case of formations with angle and signed angle constraints, if a framework $(\mathcal{G},p)$ satisfies $\rank(R^s_A)=2n-4=m_a+m_s$ in $\mathbb{R}^2$, then the number of constraints, $m_a+m_s = 2n-4$, becomes the minimal number of constraints. Moreover, an IRDS (resp. IRAS) formation with the minimal number of constraints is called minimally IRDS (resp. minimally IRAS) or, in short, minimally rigid.
If the number of constraints is greater than the minimal number of constraints while  a framework $(\mathcal{G},p)$ satisfies $\rank(R^s_D(p))=2n-3$ or $\rank(R^s_A)=2n-4$  in $\mathbb{R}^2$, then it is a non-minimal IRDS (resp. non-minimal IRAS) in $\mathbb{R}^2$. 

We next observe another ambiguous phenomenon caused by the signed constraints, which is distinct from the flip and flex ambiguity discussed in Section~\ref{Subsec:motivation}, and we call it \myemph{sine ambiguity}. This ambiguity stems from the function ambiguity in the signed constraint involving sine functions, i.e., the fact that $\text{sin}(\alpha) = \text{sin}(\pi - \alpha)$, $\forall \alpha \in [0, \pi)$.  
\myfig\ref{Fig:1} shows an example for the sine ambiguity. Although we can make a framework satisfy the Theorem \ref{Thm:Inf_Rank_dist} or Theorem \ref{Thm:Inf_Rank_ang}, we may not be sure that the framework has a unique formation shape as shown in \myfig\ref{Fig:1}.
This problem can be resolved by well chosen constraints or imposing additional constraints. For example, if a formation as shown in \myfig\ref{Fig:well_chosen01} has the signed angle constraint $\bar{S}_{213}=\sin\left(\frac{1}{2}\pi \right)$ then the triangular formation is uniquely determined up to translations and rotations. Another example is illustrated in \myfig\ref{Fig:well_chosen02} where the formation is (distance) rigid and also IRDS (by checking the rank condition of the associated rigidity matrices). 
Removing such an ambiguity may require additional number of  distance constraints, or a set of well-chosen constraints of  distance and angle variables involving both sine and cosine functions. We here suggest combining the concept of the hybrid distance-angle rigidity with the concepts of the (distance) rigidity and the weak rigidity, which can completely eliminate the sine ambiguity although cannot make a formation minimally rigid.
\begin{proposition}\label{Proposition_1}
An infinitesimally (distance) rigid framework $(\mathcal{G},p)$ with $n \ge 2$ and $\mathcal{S}$ is IRDS in $\mathbb{R}^{2}$ if and only if  $\rank(R^s_D(p))=2n-3$.  
\end{proposition}

\begin{proposition}\label{Proposition_2}
An infinitesimally weakly rigid framework $(\mathcal{G},p)$ with $\mathcal{S}$ is IRAS in $\mathbb{R}^{2}$ if and only if  $\rank(R^s_A(p))=2n-4$.  
\end{proposition}

Moreover, the sine ambiguity is also observed in $\mathbb{R}^{3}$ since the definition of the signed constraint in $\mathbb{R}^{3}$, which is defined in the following section, includes information of the sine function. A sine ambiguity in $\mathbb{R}^{3}$ can be dealt with in the same way as in $\mathbb{R}^{2}$, which is discussed in Subsection \ref{Subsec:number_const_3D.
}

\begin{figure}[]
\centering
\subfigure[$(\theta^s)^2_{13}=\sin\left(\frac{1}{3}\pi \right)$]{ \label{Fig:1_a}
\quad\,\,\begin{tikzpicture}[scale=1]
\node[place] (node2) at (0,2) [label=above:$2$] {};
\node[place] (node1) at (-1,0) [label=below:$1$] {};
\node[place] (node3) at (1,0) [label=below:$3$] {};

\draw[lineUD] (node1)  -- node [above left] {$d_{12}$} (node2);
\draw[lineUD] (node2)  -- node [above right] {$d_{23}$} (node3);
\draw[dashed] (node1)  -- (node3);

\pic [draw, -, "$(\theta^s)^2_{13}$", angle eccentricity=1.8] {angle = node1--node2--node3};
\end{tikzpicture}\quad\,\,%
}~
\subfigure[$(\theta^s)^2_{13}=\sin\left(\frac{2}{3}\pi \right)$]{ \label{Fig:1_b}
\quad\begin{tikzpicture}[scale=1]
\node[place] (node2) at (0,1) [label=above:$2$] {};
\node[place] (node1) at (-2,0) [label=below:$1$] {};
\node[place] (node3) at (2,0) [label=below:$3$] {};

\draw[lineUD] (node1)  -- node [above left] {$d_{12}$} (node2);
\draw[lineUD] (node2)  -- node [above right] {$d_{23}$} (node3);
\draw[dashed] (node1)  -- (node3);

\pic [draw, -, "$(\theta^s)^2_{13}$", angle eccentricity=1.5] {angle = node1--node2--node3};
\end{tikzpicture}\quad%
}
\caption{Triangular formations with the same set of distance and signed constraints, where $d_{12}=d_{23}$. 
Since it holds that $\sin \left(\frac{1}{3}\pi \right)=\sin \left(\frac{2}{3}\pi \right)$, the two formations have the same constraints; however, the formations are not congruent.} \label{Fig:1}
\end{figure}
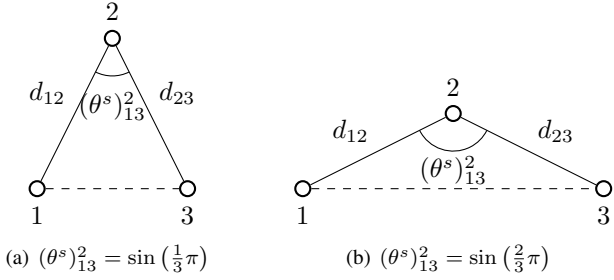

 \begin{figure}[]
\centering
\subfigure[IRDS formation with 2 distance and 1 signed constraints, where $(\theta^s)^2_{13}=\sin\left(\frac{1}{2}\pi \right)$.]{ \label{Fig:well_chosen01}
\quad\,\begin{tikzpicture}[scale=1.5]
\node[place] (node2) at (0,1) [label=above:$2$] {};
\node[place] (node1) at (-1,0) [label=below:$1$] {};
\node[place] (node3) at (1,0) [label=below:$3$] {};

\draw[lineUD] (node1)  -- node [above left] {$d_{12}$} (node2);
\draw[lineUD] (node2)  -- node [above right] {$d_{23}$} (node3);
\draw[dashed] (node1)  -- (node3);

\pic [draw, -, "$(\theta^s)^2_{13}$", angle eccentricity=1.5] {angle = node1--node2--node3};
\end{tikzpicture}\quad\,%
}\quad%
\subfigure[IRDS formation with 3 distance and 1 signed constraints.]{ \label{Fig:well_chosen02}
\quad\begin{tikzpicture}[scale=1.2]
\node[place] (node2) at (0,2) [label=above:$2$] {};
\node[place] (node1) at (-1,0) [label=below:$1$] {};
\node[place] (node3) at (1,0) [label=below:$3$] {};

\draw[lineUD] (node1)  -- node [above left] {$d_{12}$} (node2);
\draw[lineUD] (node2)  -- node [above right] {$d_{23}$} (node3);
\draw[lineUD] (node1)  -- node [below] {$d_{13}$} (node3);

\pic [draw, -, "$(\theta^s)^2_{13}$", angle eccentricity=1.8] {angle = node1--node2--node3};
\end{tikzpicture}\quad%
}%
\caption{Triangular formations with well chosen constraints. } \label{Fig:2}
\end{figure}
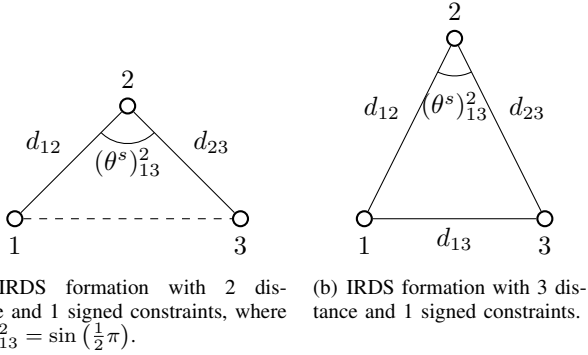

\section{Hybrid distance-angle rigidity theory in $\mathbb{R}^{3}$} \label{Sec:inf_rigid_3D}
In this section, we extend the rigidity theory in $\mathbb{R}^{2}$ developed in the last section to $\mathbb{R}^{3}$, where the signed angle constraints are substituted with signed volume constraints. That is,
we discuss the 3-dimensional case with distance (or angle) and signed-volume constraints.
The normalized signed volume is denoted by
$\bar{V}_{ijkl}=  \frac{{z_{ji}}^\top}{\norm{z_{ji}}} \left(\frac{z_{ki}}{\norm{z_{ki}}} \times \frac{z_{li}}{\norm{z_{li}}}\right)$ for edges $(i,j)$, $(i,k)$ and $(i,l)$, where the symbol $\times$ denotes a cross product.
Note that bearings are used in the definition of the signed volume.
We  define a $\mu$-th signed volume with the ordered relative position vectors as  
$\bar{V}_\mu=  \frac{{z_i}^\top}{\norm{z_i}} \left(\frac{z_j}{\norm{z_j}} \times \frac{z_k}{\norm{z_k}}\right), \mu\in \{ 1,..., m_v\}$ with the ordered relative position vectors.
We also define a set $\bar{\mathcal{S}}$ for signed volume constraints as
$\bar{\mathcal{S}} = \set{(i,j,k,l)  \in \mathcal{V}^4 \given i,j,k,l \in \mathcal{V} \text{  to characterize } \bar{V}_{ijk}}$
with $m_v=\card{\bar{\mathcal{S}}}$.
\subsection{Infinitesimal rigidity with distance and signed-volume constraints in $\mathbb{R}^{3}$}\label{Subsec:rigid_dist_3D}
We first introduce several functions to explore a concept of infinitesimal rigidity with distance and signed-volume constraints.
The \myemph{distance-volume rigidity function} $F^v_D: \chi \rightarrow \mathbb{R}^{(m_d+m_v)}$  for some properly defined $\chi\subset\mathbb{R}^{3n}$ is defined as
\begin{align}\label{eq:rigidity_fn_dist_3D}
F^v_{D}(p) \triangleq& \left[\frac{1}{2} \norm{z_{1}}^2, ... ,\frac{1}{2}\norm{z_{m_d}}^2, \bar{V}_1, ... ,\bar{V}_{m_v} \right]^\top.
\end{align}
We then consider the following time derivative of \eqref{eq:rigidity_fn_dist_3D}:
\begin{align}\label{eq:matrix_infini_dist_3D}
\dot{F}^v_D=R^v_{D}\dot{p}=0,
\end{align} 
where $R^v_{D}$ denotes the \myemph{distance-volume rigidity matrix} defined as
\begin{equation}\label{dist_sign3D_rigidity_matrix}
R^v_{D}(p) \triangleq \frac{\partial F^v_{D}(p)}{\partial p} =  \begin{bmatrix}
     \frac{\partial \mbf{D}}{\partial p}\\ \\
	\frac{\partial \mbf{V}}{\partial p}
  \end{bmatrix} = \begin{bmatrix}
     R_D\\ \\
	R_V
  \end{bmatrix} \in 
\mathbb{R}^{(m_d+m_v) \times 3n}
\end{equation}
where $\mbf{D} = \left[\frac{1}{2}\norm{z_{1}}^2, ... ,\frac{1}{2}\norm{z_{m_d}}^2\right]^\top \in \mathbb{R}^{m_d}$ and $\mbf{V} = \left[\bar{V}_1,...,\bar{V}_{m_v}\right]^\top \in \mathbb{R}^{m_v}$.
We finally reach the following results in a similar manner to Section \ref{Sec:inf_rigid}.

\begin{lemma}\label{Lem:null_rank_3d_dist}
It holds that  $\vspan\set{\mathds{1}_n\otimes I_3,(I_n\otimes {\color{black}\bar{J}_1})p,(I_n\otimes {\color{black}\bar{J}_2})p,(I_n\otimes {\color{black}\bar{J}_3})p} \subseteq \vnull(R^v_D(p))$ and $\rank\left(R^v_D(p)\right) \leq 3n-6$ for a framework $(\mathcal{G},p)$ with $\bar{\mathcal{S}}$ in $\mathbb{R}^3$, where the bases of rotational matrix  {\color{black}$\bar{J}_\sigma,\sigma=1,2,3$} are given as
{\color{black}
\begin{align} 
\bar{J}_1=\begin{bmatrix}
0 & 0 & 0 \\
0 & 0 &-1\\
0 & 1 &0
\end{bmatrix};
\bar{J}_2=\begin{bmatrix}
0 & 0 & 1 \\
0 & 0 &0\\
-1 & 0 &0
\end{bmatrix};
\bar{J}_3=\begin{bmatrix}
0 & -1 & 0 \\
1 & 0 &0\\
0 & 0 &0
\end{bmatrix}.
\end{align} 
}
\end{lemma}
\begin{proof}
First, we define a new edge set $\mathcal{E}'$ as
\begin{align}
\mathcal{E}'=\{(i,j),(i,k),(i,l) \mid (i,j)\in \mathcal{E} \lor (i,j,k,l) \in \bar{\mathcal{S}}\}
\end{align}
with $m_t=\card{\mathcal{E}'}$.
Also, the newly ordered relative position  is defined as $z'_s=z_{ij}$  for $(i,j) \in \mathcal{E}', \forall s\in \{ 1,..., m_t\}$.
Then, as discussed in the 2-dimensional case, it holds that $R_V=\frac{\partial \mbf{V}}{\partial p}=\frac{\partial \mbf{V}}{\partial z'} \frac{\partial z'}{\partial p}=\frac{\partial \mbf{V}}{\partial z'}\bar{H'}$ where $\bar{H'}=(H'\otimes I_3)$ and $z' = \big[{z'}_{1}^\top,...,{z'}_{m_t}^\top \big]^\top = (H'\otimes I_3)p \in \mathbb{R}^{3m_t}$. We also have
\begin{flalign}
\bar{H'}(I_n\otimes {\color{black}\bar{J}_\sigma})p &= (H'\otimes I_3)(I_n\otimes {\color{black}\bar{J}_\sigma})p & \nonumber \\
&=(H'\otimes {\color{black}\bar{J}_\sigma})p \nonumber \\
&= (I_{m_t} H'\otimes {\color{black}\bar{J}_\sigma} I_3)p \nonumber \\
&= (I_{m_t} \otimes {\color{black}\bar{J}_\sigma})(H'\otimes I_3)p \nonumber \\
&=(I_{m_t} \otimes {\color{black}\bar{J}_\sigma})z'= \begin{bmatrix}
{\color{black}\bar{J}_\sigma} z'_1 \\
\vdots\\
{\color{black}\bar{J}_\sigma} z'_{m_t}
\end{bmatrix}
\end{flalign}
for {\color{black}$\sigma=1,2,3$.}
Since $H'$ is an incidence matrix, it is obvious that $\vspan\set{\mathds{1}_n\otimes \begin{bmatrix} 1 & 0 & 0 \end{bmatrix}^\top} \subseteq \vnull(H'\otimes I_3)$, $\vspan\set{\mathds{1}_n\otimes \begin{bmatrix} 0 & 1 & 0  \end{bmatrix}^\top} \subseteq \vnull(H'\otimes I_3)$ and $\vspan\set{\mathds{1}_n\otimes \begin{bmatrix} 0 & 0 & 1  \end{bmatrix}^\top} \subseteq \vnull(H'\otimes I_3)$. Thus, it holds that $\vspan\set{\mathds{1}_n\otimes I_3} \subseteq \vnull(R_V(p))$.

With the fact that there exist $i$,$j$ and $k$ such that $\bar{V}_\mu=\frac{{z_i}^\top}{\norm{z_i}} \left(\frac{z_j}{\norm{z_j}} \times \frac{z_k}{\norm{z_k}}\right), \mu\in \{ 1,..., m_v\}$ and $i,j,k \in \{ 1,..., m_t\}$, we have the following derivative result
\begin{small}
\begin{align}
&\frac{\partial \bar{V}_\mu}{\partial p} (I_n\otimes {\color{black}\bar{J}_\sigma})p  =\frac{\partial \bar{V}_\mu}{\partial z'} \frac{\partial z'}{\partial p} (I_n\otimes {\color{black}\bar{J}_\sigma})p  =\frac{\partial \bar{V}_\mu}{\partial z'} \bar{H'}  (I_n\otimes {\color{black}\bar{J}_\sigma})p &\nonumber \\
&=\left(\frac{\partial}{\partial z'} \frac{{z'_i}^\top}{\norm{z'_i}} \left(\frac{z'_j}{\norm{z'_j}} \times \frac{z'_k}{\norm{z'_k}}\right)\right)\begin{bmatrix}
{\color{black}\bar{J}_\sigma} z'_1 \\
\vdots\\
{\color{black}\bar{J}_\sigma} z'_{m_t}
\end{bmatrix} &\nonumber \\
&= \left(\frac{z'_j}{\norm{z'_j}} \times \frac{z'_k}{\norm{z'_k}}\right)^\top  \left(\frac{\partial}{\partial z'_i}\frac{z'_i}{\norm{z'_i}}\right){\color{black}\bar{J}_\sigma} z'_i &\nonumber \\
&\quad- \frac{{z'_i}^\top}{\norm{z'_i}} \xi_{\frac{z'_k}{\norm{z'_k}}}\left(\frac{\partial}{\partial z'_j}\frac{z'_j}{\norm{z'_j}}\right){\color{black}\bar{J}_\sigma} z'_j &\nonumber \\
&\quad+ \frac{{z'_i}^\top}{\norm{z'_i}} \xi_{\frac{z'_j}{\norm{z'_j}}} \left(\frac{\partial}{\partial z'_k}\frac{z'_k}{\norm{z'_k}}\right){\color{black}\bar{J}_\sigma} z'_k  &\nonumber \\
&= \left(\frac{z'_j}{\norm{z'_j}} \times \frac{z'_k}{\norm{z'_k}}\right)^\top {\color{black}\bar{J}_\sigma} \frac{z'_i}{\norm{z'_i}}
- \frac{{z'_i}^\top}{\norm{z'_i}} \xi_{\frac{z'_k}{\norm{z'_k}}}{\color{black}\bar{J}_\sigma}\frac{z'_j}{\norm{z'_j}}&\nonumber \\
&\quad+ \frac{{z'_i}^\top}{\norm{z'_i}} \xi_{\frac{z'_j}{\norm{z'_j}}} {\color{black}\bar{J}_\sigma} \frac{z'_k}{\norm{z'_k}},& \label{temp_eq_3d_01}
\end{align} 
\end{small}where we have used the fact that ${z'}_i^\top {\color{black}\bar{J}_\sigma} z'_i={z'}_j^\top {\color{black}\bar{J}_\sigma} z'_j={z'}_k^\top {\color{black}\bar{J}_\sigma} z'_k=0, {\color{black}\sigma=1,2,3}$, and the Jacobian rule of cross product (see Lemma \ref{Lem:Jaco_cross} in Appendix).
Moreover, let us denote $\frac{z'_i}{\norm{z'_i}}= \begin{bmatrix}\hat{i}_1 & \hat{i}_2&\hat{i}_3\end{bmatrix}^\top\in\mathbb{R}^{3}$, $\frac{z'_j}{\norm{z'_j}}= \begin{bmatrix}\hat{j}_1 & \hat{j}_2&\hat{j}_3\end{bmatrix}^\top\in\mathbb{R}^{3}$ and $\frac{z'_k}{\norm{z'_k}}= \begin{bmatrix}\hat{k}_1 & \hat{k}_2&\hat{k}_3\end{bmatrix}^\top\in\mathbb{R}^{3}$. Then, we can calculate \eqref{temp_eq_3d_01} as follows
\begin{align}
&\left(\frac{z'_j}{\norm{z'_j}} \times \frac{z'_k}{\norm{z'_k}}\right)^\top {\color{black}\bar{J}_\sigma} \frac{z'_i}{\norm{z'_i}}
- \frac{{z'_i}^\top}{\norm{z'_i}} \xi_{\frac{z'_k}{\norm{z'_k}}}{\color{black}\bar{J}_\sigma}\frac{z'_j}{\norm{z'_j}}&\nonumber \\
&\quad+ \frac{{z'_i}^\top}{\norm{z'_i}} \xi_{\frac{z'_j}{\norm{z'_j}}} {\color{black}\bar{J}_\sigma} \frac{z'_k}{\norm{z'_k}}& \nonumber \\
&=\begin{bmatrix}\hat{j}_2 \hat{k}_3-\hat{j}_3 \hat{k}_2 & \hat{j}_3 \hat{k}_1-\hat{j}_1 \hat{k}_3&\hat{j}_1 \hat{k}_2-\hat{j}_2 \hat{k}_1\end{bmatrix} {\color{black}\bar{J}_\sigma} \begin{bmatrix}
\hat{i}_1 \\
\hat{i}_2 \\
\hat{i}_3
\end{bmatrix} &\nonumber \\
&\quad-\begin{bmatrix}\hat{i}_1 & \hat{i}_2&\hat{i}_3\end{bmatrix} 
	\begin{bmatrix}
     		0 & -\hat{k}_3 & \hat{k}_2 \\
     		\hat{k}_3 & 0 & -\hat{k}_1 \\
     		-\hat{k}_2 & \hat{k}_1 & 0
  	\end{bmatrix}
{\color{black}\bar{J}_\sigma} 
	\begin{bmatrix}
	\hat{j}_1 \\
	\hat{j}_2 \\
	\hat{j}_3
	\end{bmatrix} &\nonumber \\
&\quad+\begin{bmatrix}\hat{i}_1 & \hat{i}_2&\hat{i}_3\end{bmatrix} 
	\begin{bmatrix}
     		0 & -\hat{j}_3 & \hat{j}_2 \\
     		\hat{j}_3 & 0 & -\hat{j}_1 \\
     		-\hat{j}_2 &\hat{j}_1 & 0
  	\end{bmatrix}
{\color{black}\bar{J}_\sigma}
	\begin{bmatrix}
	\hat{k}_1 \\
	\hat{k}_2 \\
	\hat{k}_3
	\end{bmatrix} &\nonumber \\
&=0  &
\end{align} for {\color{black}$\sigma=1,2,3$.}
Thus, for all $v\in \{ 1,..., m_v\}$,  it holds that $R_V(p) p =0$, which implies that $\vspan\set{p}  \subseteq \vnull(R_V(p))$. 

Consequently, it holds that $\vspan\set{\mathds{1}_n\otimes I_3,{\color{black}(I_n\otimes \bar{J}_1)p,(I_n\otimes \bar{J}_2)p,(I_n\otimes \bar{J}_3)p}} \subseteq \vnull(R_V(p))$. 
From \cite[Lemma 1]{sun2017distributed}, one knows that  $\vspan\set{\mathds{1}_n\otimes I_3,{\color{black}(I_n\otimes \bar{J}_1)p,(I_n\otimes \bar{J}_2)p,(I_n\otimes \bar{J}_3)p}} \subseteq \vnull(R_D(p))$. Therefore, it also holds that  $\vspan\set{\mathds{1}_n\otimes I_3,{\color{black}(I_n\otimes \bar{J}_1)p,(I_n\otimes \bar{J}_2)p,(I_n\otimes \bar{J}_3)p}} \subseteq \vnull(R^v_D(p))$, which implies that $\rank\left(R^v_D(p)\right) \leq 3n-6$.
\end{proof}

\begin{theorem}
\label{Thm:Inf_Rank_dist_3D} 
A framework $(\mathcal{G},p)$ with $\bar{\mathcal{S}}$ is IRDS in $\mathbb{R}^{3}$ if and only if  $\rank(R^v_D(p))=3n-6$. 
\end{theorem}
\begin{proof}
From Lemma \ref{Lem:null_rank_3d_dist}, we have $\vspan\set{\mathds{1}_n\otimes I_3,{\color{black}(I_n\otimes \bar{J}_1)p,(I_n\otimes \bar{J}_2)p,(I_n\otimes \bar{J}_3)p}} \subseteq \vnull(R^v_D(p))$ and $\rank\left(R^v_D(p)\right) \leq 3n-6$, and further $\mathds{1}_n\otimes I_3$ and {\color{black}$(I_n\otimes \bar{J}_\sigma)p,\sigma=1,2,3$} correspond to the rigid-body translation and rotation of an entire framework in $\mathbb{R}^{3}$, respectively. Thus, the condition `$\rank(R^v_D(p))=3n-6$' means that all infinitesimal motions satisfying \eqref{eq:matrix_infini_dist_3D} are trivial and vice versa.
\end{proof}

\subsection{Infinitesimal rigidity with angle- and signed volume-constraint in that in $\mathbb{R}^{3}$}\label{Subsec:rigid_ang_3D}
This sections follows the same process as Subsection \ref{Subsec:rigid_dist_3D}.
We define the \myemph{angle-volume rigidity function} $F^v_A: \chi \rightarrow \mathbb{R}^{(m_a+m_v)}$  for some properly defined $\chi\subset\mathbb{R}^{3n}$ as follows
\begin{align} \label{eq:rigidity_fn_ang_3D}
F^v_{A}(p) \triangleq& \left[ A_1, ... ,A_{m_a}, \bar{V}_1, ... ,\bar{V}_{m_v} \right]^\top.
\end{align}
Let us consider the following derivative of \eqref{eq:rigidity_fn_ang_3D}:
\begin{align}\label{eq:matrix_infini_ang_3D}
\dot{F}^v_A=R^v_{A}\dot{p}=0,
\end{align} 
where $R^v_{A}$ is the \myemph{angle-volume rigidity matrix} defined as
\begin{equation}\label{ag_sign3D_rigidity_matrix}
R^v_{A}(p) \triangleq \frac{\partial F^v_{A}(p)}{\partial p} =  \begin{bmatrix}
     \frac{\partial \mbf{A}}{\partial p}\\ \\
	\frac{\partial \mbf{V}}{\partial p}
  \end{bmatrix} = \begin{bmatrix}
     R_W\\ \\
	R_V
  \end{bmatrix} \in 
\mathbb{R}^{(m_a+m_v) \times 3n}
\end{equation}
where $\mbf{A} = \left[A_1, ... ,A_{m_a}\right]^\top \in \mathbb{R}^{m_a}$.
We then have the following results in the same way as Section \ref{Sec:inf_rigid}.

\begin{lemma}\label{Lem:null_rank_3d_angle}
It holds that $\vspan\set{\mathds{1}_n\otimes I_3,{\color{black}(I_n\otimes \bar{J}_1)p,(I_n\otimes \bar{J}_2)p,(I_n\otimes \bar{J}_3)p},p} \subseteq \vnull(R^v_A(p))$ and $\rank\left(R^v_A(p)\right) \leq 3n-7$ for a framework $(\mathcal{G},p)$ with $\bar{\mathcal{S}}$ in $\mathbb{R}^3$. 
\end{lemma}
\begin{proof}
In the same way as the proof of Lemma \ref{Lem:null_rank_3d_dist}, it holds that $\vspan\set{\mathds{1}_n\otimes I_3,{\color{black}(I_n\otimes \bar{J}_1)p,(I_n\otimes \bar{J}_2)p,(I_n\otimes \bar{J}_3)p}} \subseteq \vnull(R_V(p))$.

Next, we define a new edge set $\mathcal{E}'$ as 
\begin{align}
\mathcal{E}'=\{(i,j),(i,k),(i,l) \mid (i,j,k) \in \mathcal{A} \lor (i,j,k,l) \in \bar{\mathcal{S}}\}.
\end{align}
with $m_t=\card{\mathcal{E}'}$.
Given the subscripts $i$,$j$ and $k$ such that $\bar{V}_\mu=\frac{{z_i}^\top}{\norm{z_i}} \left(\frac{z_j}{\norm{z_j}} \times \frac{z_k}{\norm{z_k}}\right), \mu\in \{ 1,..., m_v\}$ and $i,j,k \in \{ 1,..., m_t\}$, we have the following result with notations in the proof of Lemma \ref{Lem:null_rank_3d_dist} and reference to Lemma \ref{Lem:Jaco_cross}.
\begin{align}
\frac{\partial \bar{V}_\mu}{\partial p} p 
&=\frac{\partial \bar{V}_\mu}{\partial z'} \frac{\partial z'}{\partial p} p  =\frac{\partial \bar{V}_\mu}{\partial z'} \bar{H'} p \nonumber \\
&=\left(\frac{\partial}{\partial z'} \frac{{z'_i}^\top}{\norm{z'_i}} \left(\frac{z'_j}{\norm{z'_j}} \times \frac{z'_k}{\norm{z'_k}}\right)\right)			
	\begin{bmatrix}
	z'_1 \\
	\vdots\\
	z'_{m_t}
	\end{bmatrix} \nonumber \\
&= \left(\frac{z'_j}{\norm{z'_j}} \times \frac{z'_k}{\norm{z'_k}}\right)^\top  \left(\frac{\partial}{\partial z'_i}\frac{z'_i}{\norm{z'_i}}\right) z'_i \nonumber \\
&\quad- \frac{{z'_i}^\top}{\norm{z'_i}} \xi_{\frac{z'_k}{\norm{z'_k}}}\left(\frac{\partial}{\partial z'_j}\frac{z'_j}{\norm{z'_j}}\right) z'_j \nonumber \\
&\quad+ \frac{{z'_i}^\top}{\norm{z'_i}} \xi_{\frac{z'_j}{\norm{z'_j}}} \left(\frac{\partial}{\partial z'_k}\frac{z'_k}{\norm{z'_k}}\right) z'_k  \nonumber \\
&= 0-0+0=0
\end{align} 
where it holds that $P_{z'_i} z'_i=P_{z'_j} z'_j=P_{z'_k} z'_k=0$. Thus, for all $v\in \{ 1,..., m_v\}$,  it holds that $R_V(p) p =0$, or equivalently, that $\vspan\set{p}  \subseteq \vnull(R_V(p))$. Moreover,  it also holds that $\vspan\set{\mathds{1}_n\otimes I_3,{\color{black}(I_n\otimes \bar{J}_1)p,(I_n\otimes \bar{J}_2)p,(I_n\otimes \bar{J}_3)p}} \subseteq \vnull(R_V(p))$ as proved in Lemma~\ref{Lem:null_rank_3d_dist}.

From  \cite[Lemma 2]{kwon2018generalized}, one concludes that $\vspan\set{\mathds{1}_n\otimes I_3,{\color{black}(I_n\otimes \bar{J}_1)p,(I_n\otimes \bar{J}_2)p,(I_n\otimes \bar{J}_3)p},p} \subseteq \vnull(R_W(p))$ when $\mathcal{E} = \emptyset$. Therefore, it also holds that  $\vspan\set{\mathds{1}_n\otimes I_3,{\color{black}(I_n\otimes \bar{J}_1)p,(I_n\otimes \bar{J}_2)p,(I_n\otimes \bar{J}_3)p},p} \subseteq \vnull(R^v_A(p))$, which implies that $\rank\left(R^v_A(p)\right) \leq 3n-7$.
\end{proof}

\begin{theorem}
\label{Thm:Inf_Rank_ang_3D} 
A framework $(\mathcal{G},p)$ with $\bar{\mathcal{S}}$ is IRAS in $\mathbb{R}^{3}$ if and only if  $\rank(R^v_A(p))=3n-7$. 
\end{theorem}
\begin{proof}
With the fact that $\mathds{1}_n\otimes I_3$, {\color{black}$(I_n\otimes \bar{J}_\sigma)p,\sigma=1,2,3$} and $p$ respectively correspond to a rigid-body translation, a rotation and a scaling transformation of an entire framework, we can reach the result in a similar way to the proof of Theorem \ref{Thm:Inf_Rank_ang}.
Thus, 
we have that the equality $\rank(R^v_A(p))=3n-7$ holds if and only if all infinitesimal motions satisfying \eqref{eq:matrix_infini_ang_3D} are trivial. 
\end{proof}

{\color{black}
\subsection{Discussion on minimal rigidity with hybrid and signed constraints in $\mathbb{R}^{3}$} \label{Subsec:number_const_3D}
The minimal rigidity in $\mathbb{R}^{3}$ can be defined in the same manner as  in $\mathbb{R}^{2}$. As a result, if a framework $(\mathcal{G},p)$ satisfies $\rank(R^v_D)=3n-6=m_d+m_v$ in $\mathbb{R}^3$, then $(\mathcal{G},p)$ is minimally IRDS.
If a framework $(\mathcal{G},p)$ satisfies $\rank(R^v_A)=3n-7=m_a+m_v$ in $\mathbb{R}^3$, then $(\mathcal{G},p)$ is minimally IRAS. In addition, when a minimal rigid formation or a formation without well chosen constraints is considered, the sine ambiguity is also observed in $\mathbb{R}^{3}$, which is due to the fact that the cross product in the definition of the signed volume constraint includes information of the sine function. Thus, we again suggest combining the concept of the hybrid distance-angle rigidity with the concepts of the (distance) rigidity and the weak rigidity to completely eliminate the sine ambiguity, which cannot also make a formation minimally rigid.
\begin{proposition}\label{Proposition_3}
An infinitesimally (distance) rigid framework $(\mathcal{G},p)$ with $n \ge 3$ and $\bar{\mathcal{S}}$ is IRDS in $\mathbb{R}^{3}$ if and only if  $\rank(R^v_D(p))=3n-6$.  
\end{proposition}
\begin{proposition}\label{Proposition_4}
An infinitesimally weakly rigid framework $(\mathcal{G},p)$ with $\bar{\mathcal{S}}$ is IRAS in $\mathbb{R}^{3}$ if and only if  $\rank(R^v_A(p))=3n-7$.  
\end{proposition}
}
{\color{black}
\section{Signed Henneberg construction with hybrid distance-angle and signed constraints} \label{Sec:solution_ambiguity}
As discussed in Sections \ref{Sec:inf_rigid} and \ref{Sec:inf_rigid_3D}, the hybrid distance-angle rigidity with signed constraints can only guarantee locally unique formation shape. Moreover, the ambiguity issues cannot fully be resolved. In this sense, we would like to suggest a specific technique by modifying Henneberg construction \cite{tay1985generating,eren2004merging}, where the modified Henneberg construction is termed \myemph{signed Henneberg construction}. The traditional Henneberg construction is used to generate minimally rigid formations \cite{anderson2008rigid,eren2004merging}. However, in this paper, we use the modified Henneberg construction in a different way such that the signed Henneberg construction is a sequential technique to make a formation globally unique\footnote{{\color{black}A globally unique formation shape means that the shape cannot be deformed by any motion of the formation except a translation and a rotation of entire formation by a set of distance and signed constraints (or a translation, a rotation and a scaling of entire formation by a set of unsigned angle and signed constraints).}}. The main idea is to use triangular formations in 2D  and tetrahedral formations in 3D as shown in \myfig\ref{Fig:3_distance} and \myfig\ref{Fig:3_angle}. 
}
\begin{figure}[]
\centering
\subfigure[Triangular formation with 3-distance and 1-signed angle constraints in 2D.]{ \label{Fig:3_a}
\quad\,\,\begin{tikzpicture}[scale=1]
\node[place] (node2) at (0,2) [label=above:$2$] {};
\node[place] (node1) at (-1,0) [label=left:$1$] {};
\node[place] (node3) at (1,0) [label=right:$3$] {};

\draw[lineUD] (node1)  -- node [above left] {$d_{12}$} (node2);
\draw[lineUD] (node2)  -- node [above right] {$d_{23}$} (node3);
\draw[lineUD] (node1)  -- node [below] {$d_{13}$} (node3);

\pic [draw, -, "$(\theta^s)^2_{13}$", angle eccentricity=1.8] {angle = node1--node2--node3};
\end{tikzpicture}\quad\,\,%
}~
\subfigure[Tetrahedral formation with 6-distance and 1-signed volume constraints in 3D, where the volume constraint is $\bar{V}_{4123}$.]{ \label{Fig:3_b}
\quad\,\,\begin{tikzpicture}[scale=0.7]
\node[place] (node1) at (0.5,0,3.5) [label=below:$1$] {};
\node[place] (node2) at (3.5,0,4.5) [label=below:$2$] {};
\node[place] (node3) at (3,0,1) [label=right:$3$] {};
\node[place] (node4) at (0,1.8,0) [label=above right:$4$] {};

\draw (node1) [lineUD] -- node [left] {} (node2);
\draw (node1) [dashed] -- node [right] {} (node3);
\draw (node1) [lineUD] -- node [right] {} (node4);
\draw (node2) [lineUD] -- node [below] {} (node3);
\draw (node2) [lineUD] -- node [below] {} (node4);
\draw (node3) [lineUD] -- node [below] {} (node4);
\end{tikzpicture}\quad\,\,%
}
\caption{{\color{black}Formations with distance and signed angle/volume constraints}} \label{Fig:3_distance}
\subfigure[Triangular formation with 2-angle and 1-signed angle constraints in 2D.]{ \label{Fig:3_c}
\quad\begin{tikzpicture}[scale=1]
\node[place] (node2) at (0,2) [label=above:$2$] {};
\node[place] (node1) at (-1,0) [label=left:$1$] {};
\node[place] (node3) at (1,0) [label=right:$3$] {};

\draw[dashed] (node1)  --  (node2);
\draw[dashed] (node2)  --  (node3);
\draw[dashed] (node1)  --  (node3);

\pic [draw, -, "$(\theta^s)^2_{13}$", angle eccentricity=1.8] {angle = node1--node2--node3};
\pic [draw, -, "$\theta^1_{23}$", angle eccentricity=1.6] {angle = node3--node1--node2};
\pic [draw, -, "$\theta^3_{12}$", angle eccentricity=1.6] {angle = node2--node3--node1};
\end{tikzpicture}\quad
}~
\subfigure[Tetrahedral formation with 5-angle and 1-signed volume constraints in 3D, where the volume constraint is $\bar{V}_{4123}$.]{ \label{Fig:3_d}
\quad\,\,\begin{tikzpicture}[scale=0.7]
\node[place] (node1) at (0.5,0,3.5) [label=below:$1$] {};
\node[place] (node2) at (3.5,0,4.5) [label=below:$2$] {};
\node[place] (node3) at (3,0,1) [label=right:$3$] {};
\node[place] (node4) at (0,1.8,0) [label=above right:$4$] {};

\draw (node1) [dashed] -- node [left] {} (node2);
\draw (node1) [dashed] -- node [right] {} (node3);
\draw (node1) [dashed] -- node [right] {} (node4);
\draw (node2) [dashed] -- node [below] {} (node3);
\draw (node2) [dashed] -- node [below] {} (node4);
\draw (node3) [dashed] -- node [below] {} (node4);

\pic [draw, -,angle radius=0.6cm] {angle = node2--node1--node3};
\pic [draw, -,angle radius=0.4cm] {angle = node1--node3--node2};
\pic [draw, -,angle radius=0.4cm] {angle = node4--node2--node1};
\pic [draw, -,angle radius=0.4cm] {angle = node1--node4--node2};
\pic [draw, -,angle radius=0.6cm] {angle = node2--node4--node3};
\end{tikzpicture}\quad\,\,%
}
\caption{{\color{black}Formations with unsigned angle and signed angle/volume constraints.}} \label{Fig:3_angle}
\end{figure}

{\color{black}
 The operation of the signed Henneberg construction is as follows. For a given globally unique formation  $\mathcal{G} = (\mathcal{V},\mathcal{E}, \mathcal{A})$ in 2D, a vertex $i$, one signed angle constraint and two distance constraints (or two unsigned angle constraints in place of the two distance constraints) are added to the formation in order that the combined formation is composed of only triangular formations with hybrid distance-angle and signed constraints; for example, see \myfig\ref{Fig:4_Henneberg}. In the case of 3D, a vertex $i$, one signed volume constraint and three distance constraints (or three unsigned angle constraints in place of the three distance constraints) are added to a given globally unique formation  $\mathcal{G} = (\mathcal{V},\mathcal{E}, \mathcal{A})$ in order that the combined formation is composed of only tetrahedral formations with hybrid distance-angle and signed constraints.  
 We can see that a formation generated by sequences of the signed Henneberg construction is globally unique since any motion of the formation is not allowed except  a translation and a rotation of entire formation by a set of distance and signed constraints (or a translation, a rotation and a scaling of entire formation by a set of unsigned angle and signed constraints).
 We then have the following result with Propositions \ref{Proposition_1}, \ref{Proposition_2}, \ref{Proposition_3} and \ref{Proposition_4}.
 
\begin{corollary}
If an IRDS or IRAS framework $(\mathcal{G},p)$  is generated by sequences of the signed Henneberg construction, then the framework is globally unique up to a translation and a rotation by a set of distance and signed constraints or a translation, a rotation and a scaling by a set of unsigned angle and signed constraints.
\end{corollary}
}
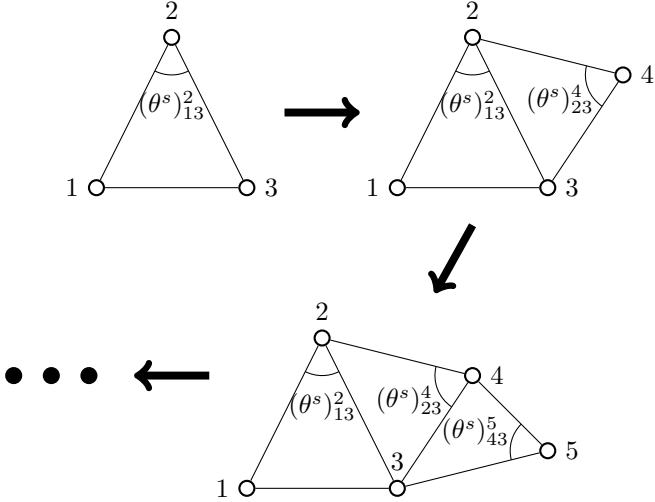
\begin{figure}[]
\centering
\begin{tikzpicture}[scale=1]
\node[place] (node2) at (0,2) [label=above:$2$] {};
\node[place] (node1) at (-1,0) [label=left:$1$] {};
\node[place] (node3) at (1,0) [label=right:$3$] {};

\draw[lineUD] (node1)  -- node{} (node2);
\draw[lineUD] (node2)  -- node{} (node3);
\draw[lineUD] (node1)  -- node{} (node3);

\pic [draw, -, "$(\theta^s)^2_{13}$", angle eccentricity=1.8] {angle = node1--node2--node3};
\draw[draw=black, line width=3.0pt, ->] (1.5,1) -- (2.5,1){};
\node[place] (node2') at (4,2) [label=above:$2$] {};
\node[place] (node1') at (3,0) [label=left:$1$] {};
\node[place] (node3') at (5,0) [label=right:$3$] {};
\node[place] (node4') at (6,1.5) [label=right:$4$] {};

\draw[lineUD] (node1')  -- node{} (node2');
\draw[lineUD] (node2')  -- node{} (node3');
\draw[lineUD] (node1')  -- node{} (node3');
\draw[lineUD] (node2')  -- node{} (node4');
\draw[lineUD] (node3')  -- node{} (node4');

\pic [draw, -, "$(\theta^s)^2_{13}$", angle eccentricity=1.8] {angle = node1'--node2'--node3'};
\pic [draw, -, "$(\theta^s)^4_{23}$", angle eccentricity=1.8] {angle = node2'--node4'--node3'};
\draw[draw=black, line width=3.0pt, ->] (4,-0.5) -- (3.5,-1.4){};
\node[place] (node2'') at (2,-2) [label=above:$2$] {};
\node[place] (node1'') at (1,-4) [label=left:$1$] {};
\node[place] (node3'') at (3,-4) [label=above:$3$] {};
\node[place] (node4'') at (4,-2.5) [label=right:$4$] {};
\node[place] (node5'') at (5,-3.5) [label=right:$5$] {};

\draw[lineUD] (node1'')  -- node{} (node2'');
\draw[lineUD] (node2'')  -- node{} (node3'');
\draw[lineUD] (node1'')  -- node{} (node3'');
\draw[lineUD] (node2'')  -- node{} (node4'');
\draw[lineUD] (node3'')  -- node{} (node4'');
\draw[lineUD] (node4'')  -- node{} (node5'');
\draw[lineUD] (node3'')  -- node{} (node5'');

\pic [draw, -, "$(\theta^s)^2_{13}$", angle eccentricity=1.8] {angle = node1''--node2''--node3''};
\pic [draw, -, "$(\theta^s)^4_{23}$", angle eccentricity=1.8] {angle = node2''--node4''--node3''};
\pic [draw, -, "$(\theta^s)^5_{43}$", angle eccentricity=2.0] {angle = node4''--node5''--node3''};
\draw[draw=black, line width=3.0pt, ->] (0.5,-2.5) -- (-0.5,-2.5){};
\node[vertex] (node3''') at (-1.1,-2.5) {};
\node[vertex] (node3''') at (-1.6,-2.5) {};
\node[vertex] (node3''') at (-2.1,-2.5) {};
\end{tikzpicture}
\caption{{\color{black}Example of a formation with the signed Henneberg construction, where the solid lines denote the distance constraints.}} \label{Fig:4_Henneberg}
\end{figure}
\section{Formation shape control with hybrid distance-angle rigidity theory} \label{Sec:controller}
In this section, as an application of the hybrid rigidity theories developed in the previous sections, we solve the following formation shape stabilization problem with a set of distance (or angle) and signed constraints in $\mathbb{R}^{d}$.
\begin{problem}
{\color{black}For a given set of formations congruent to a target IRDS (or similar to a target IRAS) formation}, design a distributed control law for $n$ agents with displacement measurements of each agent's neighbors such that $e \to 0$ as  $t\to \infty$, where the symbol $e$ denotes an error which will be defined in the next subsections.
\end{problem}

We utilize the gradient control law studied in \cite{sakurama2015distributed,krick2009stabilisation,bishop2015distributed,park2014stability}. We firstly show that an IRDS formation converges to a desired formation under the proposed gradient control law. Then, we show similar results for IRAS formations. 

\subsection{Formation control with distance and signed constraints}
This subsection proposes a distributed control law with distance and signed constraints. We first define several notations.
Let us consider the following distance and signed angle constraints. 
\begin{align}
d_c(p) &= \left[\ldots, \frac{1}{2}\norm{z_{g}}^2, \ldots \right]^\top, \\
s_c(p) &= \left[\ldots, k\bar{S}_w, \ldots \right]^\top,
\end{align}
and we define vectors composed of desired distance and signed angle constraints as, respectively,
\begin{align} 
d_c^* &= \left[\ldots, \frac{1}{2}\norm{z^*_g}^2, \ldots \right]^\top, \\
s_c^* &= \left[\ldots, k\bar{S}_w^*, \ldots \right]^\top,
\end{align}
where $k$ denotes a constant.
If  signed volume constraints  are considered, then the vectors for the constraints and desired values are respectively given by
\begin{align}
s_c(p) = \left[\ldots, k\bar{V}_\mu, \ldots \right]^\top, \\
s_c^* = \left[\ldots, k\bar{V}_\mu^*, \ldots \right]^\top,
\end{align}
Then, we define the formation errors as
\begin{equation}
e(p)=\left[d_c(p)^\top, s_c(p)^\top \right]^\top - \left[d_c^{*\top}, s_c^{*\top} \right]^\top.
\end{equation}
We can now represent the following gradient-based formation controller:
\begin{align}\label{control_law_dist01}
\dot{p}=u \triangleq-\left(\nabla\phi\right)^\top&=
\begin{cases}
-{\overline{R}^s_D}^\top(p) e(p) & \text{in 2D,} \\
-{\overline{R}^v_D}^\top(p) e(p) & \text{in 3D}
\end{cases}\nonumber \\
&=-\overline{R}_D^\top(p) e(p)
\end{align}
where $\phi=\frac{1}{2}e^\top(p)e(p)$,  $\overline{R}^s_D=\begin{bmatrix}R_D^\top & k R_S^\top\end{bmatrix}^\top$ and $\overline{R}^v_D=\begin{bmatrix}R_D^\top & k R_V^\top\end{bmatrix}^\top$. The control gain $k$ is for adjusting the convergence speed, i.e., a larger $k$ can lead to a faster convergence rate. Note that $\vnull(\overline{R}^s_D)=\vnull(R^s_D)$ and $\vnull(\overline{R}^v_D)=\vnull(R^v_D)$.
This paper only assumes that each agent can measure relative positions of its neighbors, that is, no communication among agents is considered. However, angle and signed constraints require to modify the definitions of neighbors in terms of $\mathcal{G} = (\mathcal{V},\mathcal{E}, \mathcal{A})$ with $\mathcal{S}$ or $\bar{\mathcal{S}}$. 
We thus define a sensing topology based on the controllers \eqref{control_law_dist01} and \eqref{control_law_ang01} as follows.
\begin{definition}[Sensing topology]
The sensing topology follows a new undirected graph $\mathcal{G}^* = (\mathcal{V},\mathcal{E}^*)$ where $\mathcal{E}^*=\set{(i,j),(i,k),(j,k) \given (i,j) \in \mathcal{E} \lor (i,j,k) \in \mathcal{A} \lor (i,j,k) \in \mathcal{S} }$ in 2D or $\mathcal{E}^*=\set{(i,j),(i,k),(i,l),(j,k),(j,l),(k,l) \given (i,j) \in \mathcal{E} \lor (i,j,k) \in \mathcal{A} \lor (i,j,k,l) \in \bar{\mathcal{S}} }$ in 3D.
\end{definition}
\begin{theorem}\label{Thm:stability_dist}
Let $\psi$ denote a set {\color{black}of $p$  corresponding to formations} congruent to a desired IRDS formation. Then, there exists a neighborhood $\mathcal{B}_{p^*}$ of $p^*$ for any $p^* \in \psi$ such that the initial formation {\color{black} with} $p(0)$ converges to a fixed formation {\color{black} with $p^\dagger \in \psi$} exponentially fast if $p(0) \in \mathcal{B}_{p^*}$. 
\end{theorem}
\begin{proof}
This proof is based on the center manifold theory \cite{wiggins2003introduction,summers2011control}.
Let $f(p)=\eqref{control_law_dist01}=-\overline{R}_D^\top(p) e(p)$. 
Then, we have the Jacobian $H_f$ of $f(p)$ at $p^*$ as follows.
\begin{align}
H_f(p^*)&=\left.\frac{\partial f(p)}{\partial p}\right|_{p=p^*} \nonumber \\
&=-\left.\frac{\partial\overline{R}_D^\top(p)}{\partial p}\right|_{p=p^*} e(p^*)-\overline{R}_D^\top(p^*) \left.\frac{\partial e(p)}{\partial p}\right|_{p=p^*} \nonumber \\
&= -\overline{R}_D^\top(p^*) \overline{R}_D(p^*).
\end{align}
Thus, from Theorem \ref{Thm:Inf_Rank_dist} and Theorem \ref{Thm:Inf_Rank_dist_3D}, $H_f$ at $p^*$ has $d(d+1)/2$ eigenvalues with zero real part and $dn-d(d+1)/2$ eigenvalues with negative real part. 
Employing Taylor series for $f(p)$ about $p^*$, we have $f(p)=f(p^*)+\left.\frac{\partial f(p)}{\partial p}\right|_{p=p^*}\delta+g(\delta)$, where $\delta=p-p^*$.
Moreover, there exists a similarity transformation $Q$ such that $Q H_f(p^*) Q^\top=\text{diag}(\bar{H}_0,\bar{H}_n)$ where $\bar{H}_0$ is a diagonal matrix including zero eigenvalues and $\bar{H}_n$ is a square matrix having negative eigenvalues, i.e., $\bar{H}_n$ is Hurwitz. Then, defining $(\varphi^\top,\rho^\top)^\top=Q\delta$, we locally have the following form from \eqref{control_law_dist01}.
\begin{align}\label{Eq:dynamics_trs}
\dot{\varphi}&=\bar{H}_0\varphi + g_1(\varphi,\rho)\nonumber \\
&=g_1(\varphi,\rho), \nonumber \\
\dot{\rho}&=\bar{H}_n \rho + g_2(\varphi,\rho),
\end{align}
where $(g_1^\top,g_2^\top)^\top=Qg(\delta)$, and $g_1(0,0)=0$ and $g_2(0,0)=0$ since $p^*$ is an equilibrium point of \eqref{control_law_dist01}.
It holds that $J_g(p^*)=\left.\frac{\partial g(\delta)}{\partial p}\right|_{\delta=0}=0$ since $g(\delta)=f(p)-H_f(p^*)\delta$, which implies that $J_{g_1}(0,0)=0$ and $J_{g_2}(0,0)=0$ where $J_{g_1}$ and $J_{g_2}$ denote Jacobians of $g_1$ and $g_2$, respectively.
Then, with reference to Theorem 4 in \cite{summers2011control}, there exists a $\mathcal{C}^r$ (i.e., $r$ times differentiable) center manifold $\mathcal{M}$ for \eqref{Eq:dynamics_trs} with a local representation function $h(\varphi)=\rho: \mathbb{R}^{d(d+1)/2} \rightarrow \mathbb{R}^{dn-d(d+1)/2}$, and the dynamics \eqref{Eq:dynamics_trs} on the center manifold is governed by $d(d+1)/2$-dimensional nonlinear system:
\begin{align}
\dot{\xi}&=g_1(\xi,h(\xi)) 
\end{align}
for sufficiently small $\xi \in \mathbb{R}^{d(d+1)/2}$. 

Let us next consider $d(d+1)/2$-dimensional manifold defined as $\mathcal{M}=\set{(\varphi^\top,\rho^\top)^\top\given (\varphi^\top,\rho^\top)^\top=Q\delta, p \in \psi}$. Then, $\mathcal{M}$ is invariant since $\mathcal{M}$ is an equilibrium manifold. At equilibrium points,  it holds that $\dot{\rho}=\bar{H}_n \rho + g_2(\varphi,\rho)=0$, which implies that there exists a neighborhood $B_{(\varphi,\rho)}$ of $(\varphi(0),\rho(0))$ such that $\mathcal{M}\cap B_{(\varphi,\rho)}=\set{(\varphi,\rho)\given \rho=h(\varphi), h(0)=0,J_h(0)=0}$ by the implicit function theorem, where $J_h$ is Jacobian of $h$. This means that $\mathcal{M}$ is a center manifold by the definition in \cite{wiggins2003introduction}. Therefore, $\dot{\xi}=0$ and it follows from \cite[Theorem 4]{summers2011control} that
$\varphi(t)=\xi(0)+O(exp(-\gamma t))$ and $\rho(t)=h(\xi(0))+O(exp(-\gamma t))$, where $\gamma{\color{black}>0}$ denotes a convergence {\color{black}rate}. Hence, we have the statement.
\end{proof}

\subsection{Formation control with angle and signed constraints}
In this subsection, we propose a distributed control law with angle and signed constraints. Let us first define several notations.
We denote angle and signed angle constraints by
\begin{align}
a_c(p) &= \left[\ldots, A_h, \ldots \right]^\top, \\
s_c(p) &= \left[\ldots, k\bar{S}_w, \ldots \right]^\top,
\end{align}
and desired vectors composed of angle and signed angle constraints by
\begin{align} 
a_c^* &= \left[\ldots, A^*_h, \ldots \right]^\top, \\
s_c^* &= \left[\ldots, k\bar{S}_w^*, \ldots \right]^\top,
\end{align}
where $k$ is a constant.
When the signed volume constraints   are considered, the vectors for the signed constraints and desired values are given by
\begin{align}
s_c(p) = \left[\ldots, k\bar{V}_\mu, \ldots \right]^\top, \\
s_c^* = \left[\ldots, k\bar{V}_\mu^*, \ldots \right]^\top,
\end{align}
We then define the error as
\begin{equation}
e(p)=\left[a_c(p)^\top, s_c(p)^\top \right]^\top - \left[a_c^{*\top}, s_c^{*\top} \right]^\top.
\end{equation}
Then, the gradient-based control law is defined by
\begin{align}\label{control_law_ang01}
\dot{p}=u &\triangleq
\begin{cases}
-{\overline{R}^s_A}^\top(p) e(p) & \text{in 2D,}\\ 
-{\overline{R}^v_A}^\top(p) e(p) & \text{in 3D}
\end{cases}\nonumber \\
&=-\overline{R}_A^\top(p) e(p)
\end{align}
where $\overline{R}^s_A=\begin{bmatrix}R_W^\top & k R_S^\top\end{bmatrix}^\top$ and $\overline{R}^v_A=\begin{bmatrix}R_W^\top & k R_V^\top\end{bmatrix}^\top$. Note that $\vnull(\overline{R}^s_A)=\vnull(R^s_A)$ and $\vnull(\overline{R}^v_A)=\vnull(R^v_A)$ since $k$ is a constant.

The following theorem is also one of the main results for the proposed formation control in terms of angle and signed constraints.
\begin{theorem}\label{Thm:stability_ang}
Let $\psi$ denote a set {\color{black}of $p$ corresponding to formations similar} to a desired IRAS formation. Then, there exists a neighborhood $\mathcal{B}_{p^*}$ of $p^*$ for any $p^* \in \psi$ such that the initial formation {\color{black}with} $p(0)$ converges to a fixed formation {\color{black} with $p^\dagger \in \psi$} exponentially fast if $p(0) \in \mathcal{B}_{p^*}$. 
\end{theorem}
\begin{proof}
This theorem can be proved in the same way as that of Theorem \ref{Thm:stability_dist}, and the details are omitted.
\end{proof}

\section{Simulation examples} \label{Sec:simul}
We provide two examples to validate the proposed controllers studied in Section \ref{Sec:controller}, where the final position of each agent is represented by a symbol $\Box$, and
the desired distance, signed angle and signed volume constraints are denoted by $\norm{z_{ij}^*}^2, (i,j) \in \mathcal{E}$,  $\bar{S}_{ijk}^*, (i,j,k) \in \mathcal{S}$ and $\bar{V}_{ijkl}^*, (i,j,k,l) \in \mathcal{\bar{S}}$, respectively.
The first example is to show that the proposed controller can avoid the flip ambiguity in $\mathbb{R}^{2}$, which is described in \myfig\ref{simulation02}. 
The formations for \myfig\ref{simul02_a} and \myfig\ref{simul02_b} at the initial time are infinitesimally (distance) rigid and IRDS, respectively, and the initial formation shapes for both figures are the same. In particular, three distance constraints, i.e.,$d_{23},d_{34},d_{45}$, in \myfig\ref{simul02_a} are replaced by three signed angle constraints, which depicts the formation in \myfig\ref{simul02_b} for the proposed controller \eqref{control_law_dist01}. Then, as shown in \myfig\ref{simulation02}, if the desired formation is the final formation in \myfig\ref{simul02_b} then one can observe that 
the proposed control system \eqref{control_law_dist01} does not lead to  flip ambiguity whereas the traditional control system \eqref{control_law_origin01} does. Similarly   to the results of  \myfig\ref{simulation02} in $\mathbb{R}^{2}$,
the second example in \myfig\ref{simulation03} with two signed volume constraints also presents formation convergences in $\mathbb{R}^{3}$, where the additional signed volume constraints prevent the overall formation from converging to a flipped formation shape. 
\begin{figure}[]
\centering
\subfigure[Formation control under the controller \eqref{control_law_origin01}. The desired constraints are chosen as $\norm{z_{12}^*}^2=\norm{z_{15}^*}^2=9$, $\norm{z_{23}^*}^2=\norm{z_{45}^*}^2=16$, $\norm{z_{13}^*}^2=\norm{z_{14}^*}^2=25$ and $\norm{z_{34}^*}^2=50$.]{ \label{simul02_a}
\quad \includegraphics[height = 6cm]{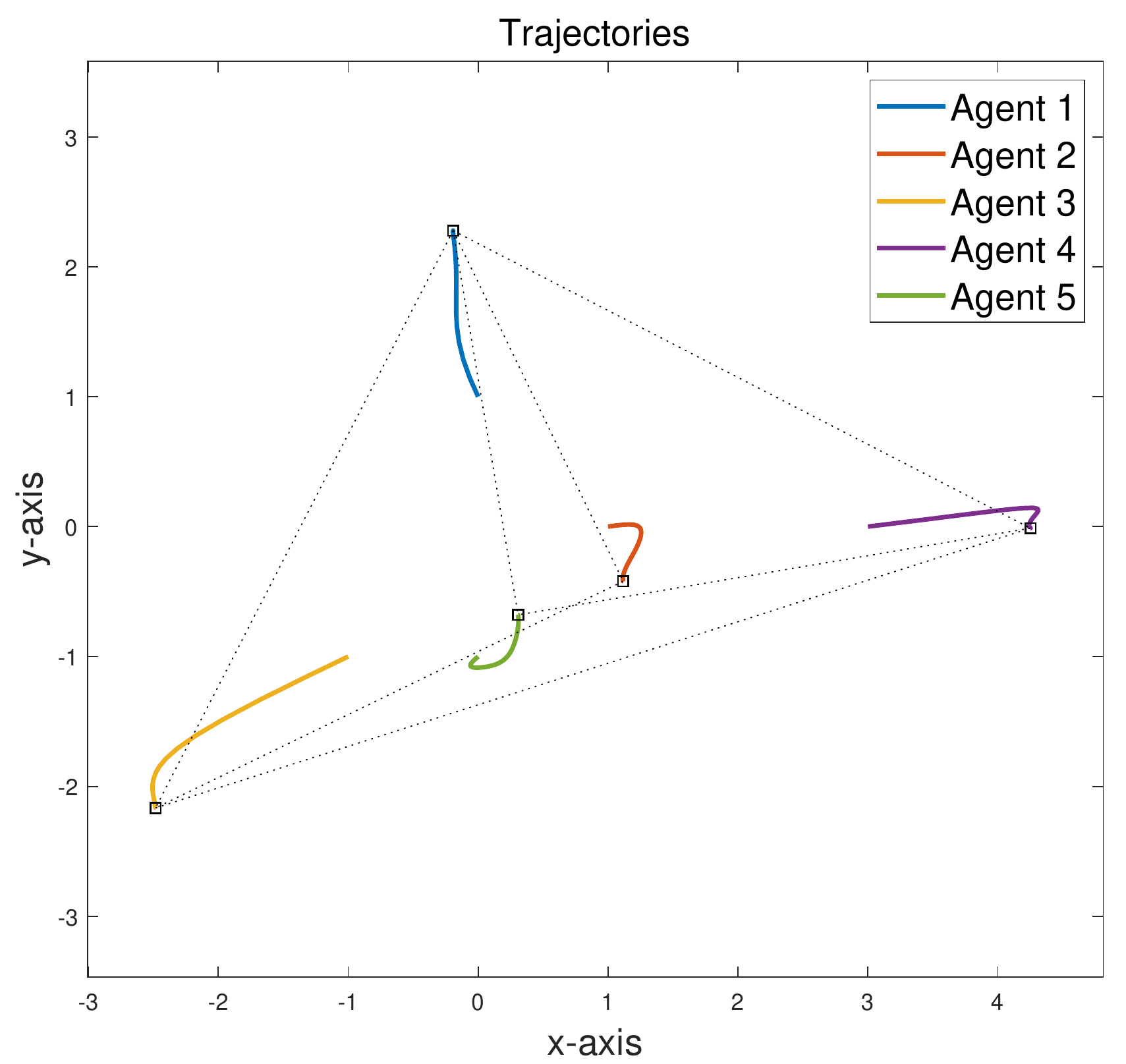} \quad
}\qquad\quad
\subfigure[Formation control under the proposed controller \eqref{control_law_dist01}. The desired constraints are chosen as $\norm{z_{12}^*}^2=\norm{z_{15}^*}^2=9$, $\norm{z_{13}^*}^2=\norm{z_{14}^*}^2=25$ and $\bar{S}_{231}^*=\bar{S}_{134}^*=\bar{S}_{514}^*=\sin\left(\frac{\pi}{2} \right)$, and the coefficient $k$ is chosen by $k=10$.]{ \label{simul02_b}
\quad\quad \includegraphics[height = 6cm]{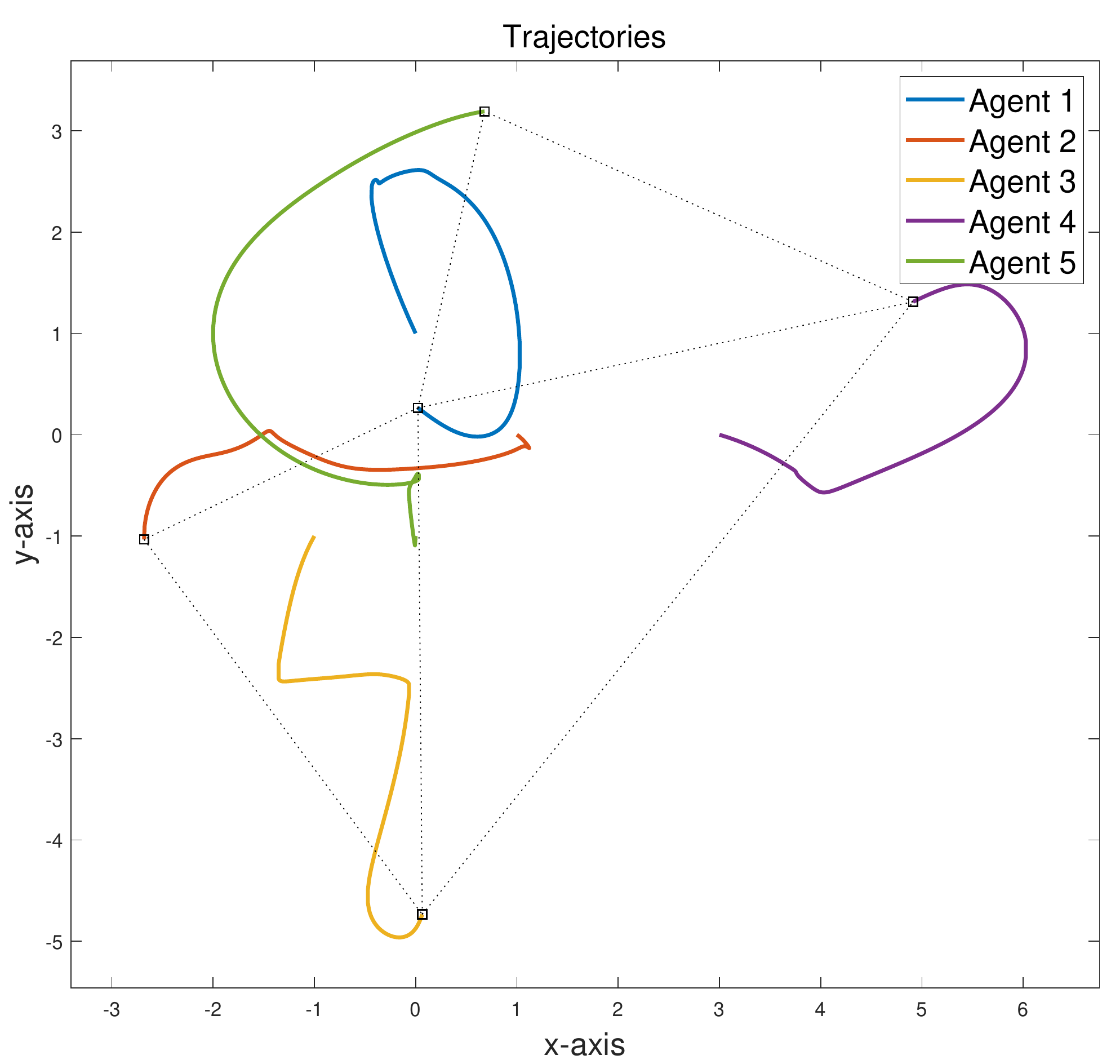} \quad\quad
}
\caption{Numerical simulations in $\mathbb{R}^{2}$ on trajectories of 5 agents from initial formation to final formation.} \label{simulation02}
\end{figure}
\begin{figure}[]
\centering
\subfigure[Formation control under the controller \eqref{control_law_origin01} with $\norm{z_{ij}^*}^2=9,(i,j) \in \mathcal{E}$.]{ \label{simul03_a}
\,\includegraphics[height = 6.5cm]{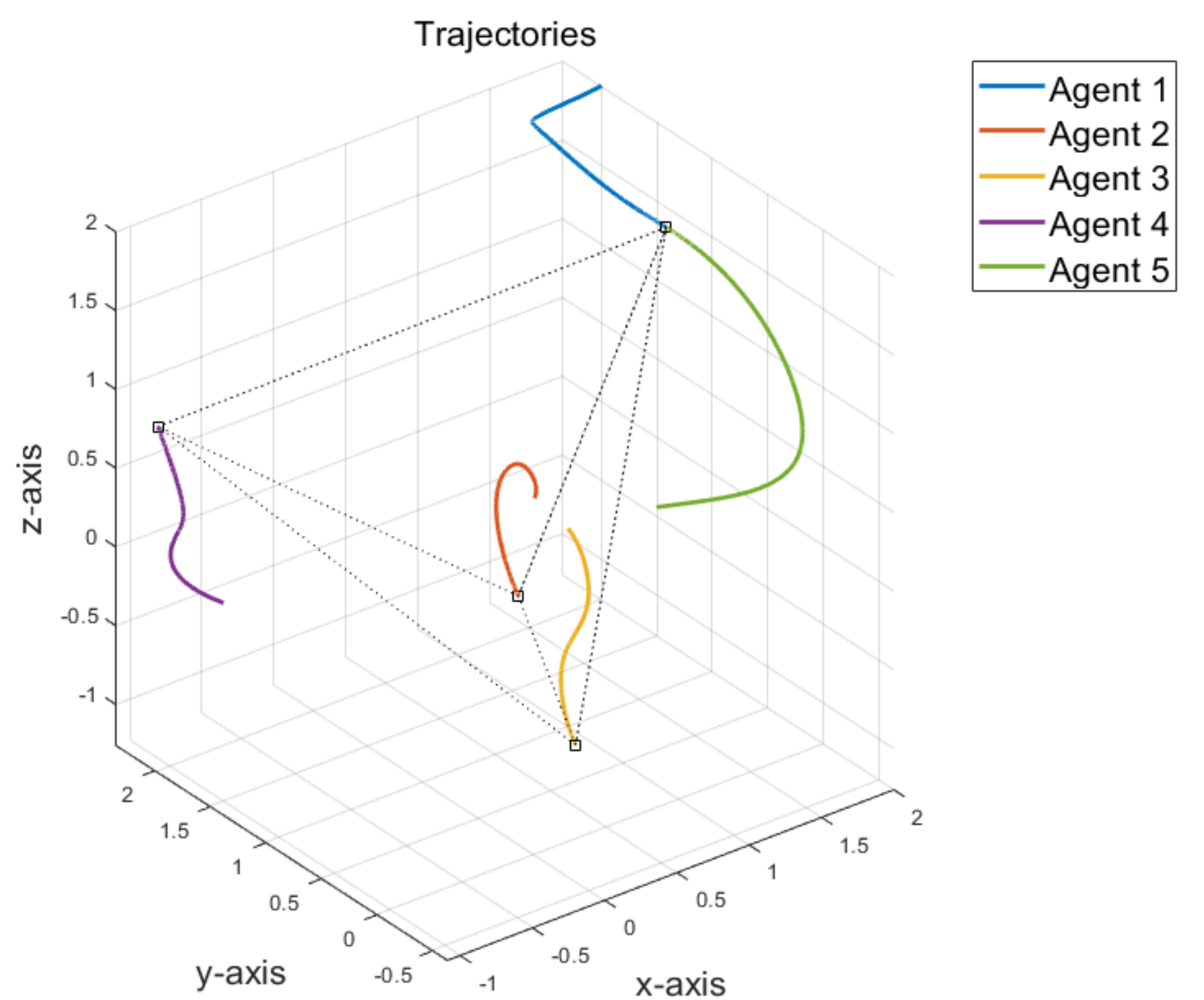}\,
}\qquad
\subfigure[Formation control under the proposed controller \eqref{control_law_dist01} with $\norm{z_{ij}^*}^2=9,(i,j) \in \mathcal{E}$, $\bar{V}_{2134}^*=\sqrt{2}/2$, $\bar{V}_{2534}^*=-\sqrt{2}/2$ and $k=10$.]{ \label{simul03_b}
\,\includegraphics[height = 6.5cm]{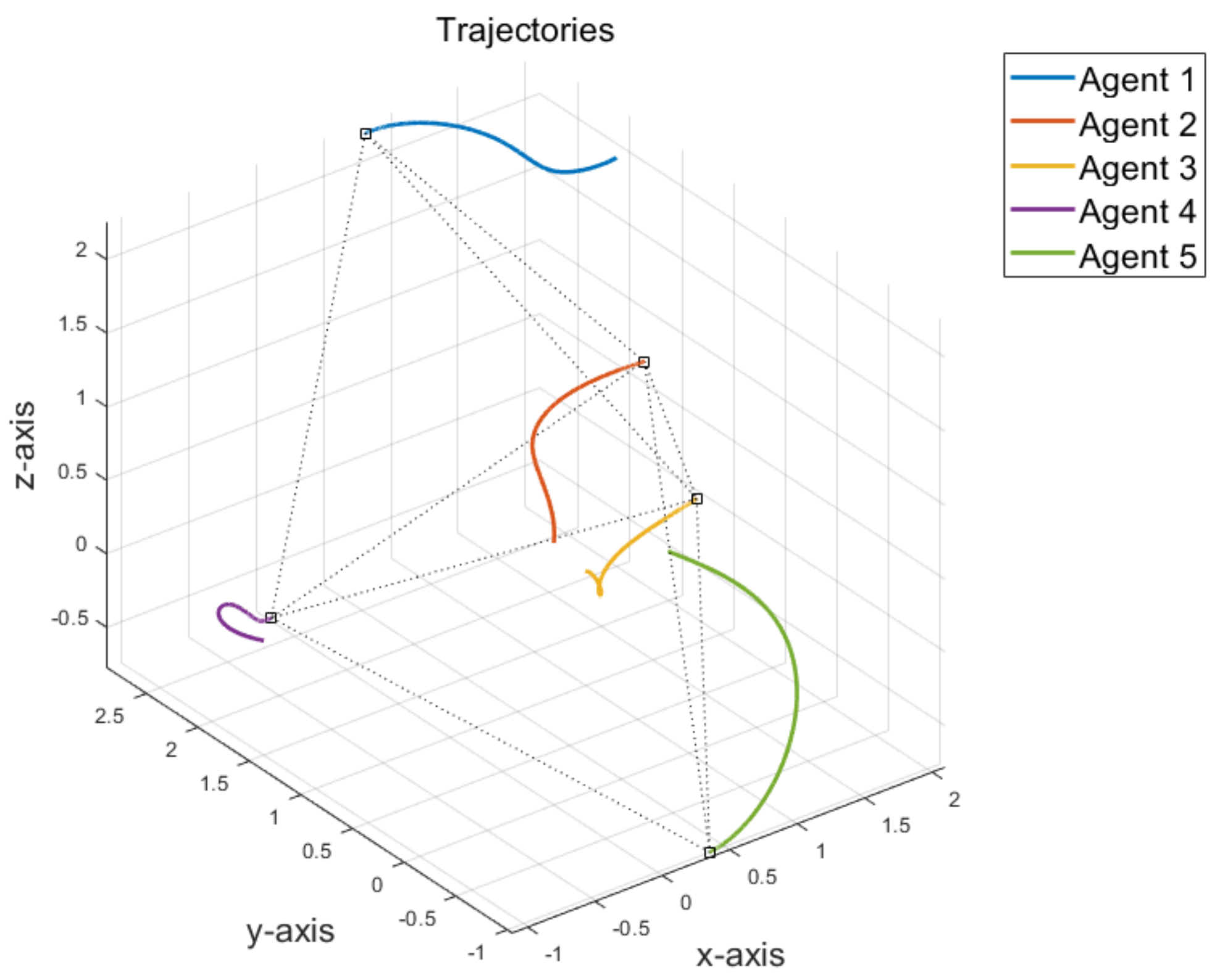} \,
}
\caption{Numerical simulations in $\mathbb{R}^{3}$ on trajectories of 5 agents from initial formation to final formation, where the edges are denoted by the dotted lines and $m_d=9$.} \label{simulation03}
\end{figure}
\section{Conclusions} \label{Sec:conclusions}
This paper develops the concepts of IRDS and IRAS for 2-D and 3-D formations with hybrid distances (or angles) and signed constraints. The signed constraints involve signed normalized area (resp. volume) functions of a given framework in the 2-D (resp. 3-D) space. By characterizing trivial infinitesimal rigid motions in terms of translations, rotations and shape scaling, we provide certain rank conditions for the associated rigidity matrix to determine the IRDS and IRAS property for a given framework. 
{\color{black}
Based on the hybrid rigidity theory, we partially resolve the flip and flex ambiguity with well chosen constraints. In particular, if IRDS or IRAS formations are generated by the signed Henneberg construction, then the ambiguity issues are completely eliminated.
}
We also apply the hybrid rigidity theory   to formation shape control. Distributed gradient-based controllers with inter-agent relative measurements are developed to stabilize   a desired IRDS (or IRAS) formation that achieves a local exponential convergence. Compared to the traditional controllers with only distance constraints {\color{black}or only unsigned angle constraints}, the proposed controllers have the property to eliminate the flip and/or flex ambiguity on $n$-agent formations {\color{black}close to a target formation} in the 2-D/3-D space.
\section{Appendix} \label{Sec:Appendix}
{\color{black}
\begin{lemma}[Jacobian of sine functions]
\label{Lem:Jaco_sine}
Consider the sine function $\left(\frac{{z}_{i}^\top}{\norm{z_{i}}}J\frac{z_{j}}{\norm{z_{j}}}\right)$ with the ordered relative positions $z_{i}$ and $z_{j}$ in $\mathbb{R}^2$. We then have the following results:
\begin{equation}
\frac{\partial}{\partial z_i} \left(\frac{{z}_i^\top}{\norm{z_i}}J\frac{z_j}{\norm{z_j}}\right)= \frac{z^\top_j}{\norm{z_j}}J^\top \frac{P_{z_i}}{\norm{z_i}}.
\end{equation}
and
\begin{equation}
\frac{\partial}{\partial z_j}\left(\frac{{z}_i^\top}{\norm{z_i}}J\frac{z_j}{\norm{z_j}}\right)=\frac{{z}_i^\top}{\norm{z_i}}J  \frac{P_{z_j}}{\norm{z_j}}
\end{equation}
\end{lemma}
\begin{proof} 
The Jacobian of a unit vector $\frac{{z}_i}{\norm{z_i}}$ is calculated as
\begin{equation}
\frac{\partial}{\partial z_i}\frac{{z}_i}{\norm{z_i}}=\frac{1}{\norm{z_i}}\left(I_d - \frac{z_i z_i^\top}{\norm{z_i}\norm{z_i}} \right)=\frac{P_{z_i}}{\norm{z_i}}.
\end{equation}
We also have 
\begin{align}
\frac{\partial}{\partial z_i} \left(\frac{{z}_i^\top}{\norm{z_i}}J\frac{z_j}{\norm{z_j}}\right)&= \frac{z^\top_j}{\norm{z_j}}J^\top \frac{\partial}{\partial z_i} \frac{{z}_i}{\norm{z_i}} \nonumber \\
\frac{\partial}{\partial z_j} \left(\frac{{z}_i^\top}{\norm{z_i}}J\frac{z_j}{\norm{z_j}}\right)&= \frac{z^\top_i}{\norm{z_i}}J \frac{\partial}{\partial z_j} \frac{{z}_j}{\norm{z_j}}.
\end{align}
Therefore, the proof is completed.
\end{proof}
}
\begin{lemma}[Jacobian of cross products]
\label{Lem:Jaco_cross} 
Consider the cross product $\left(\frac{z_j}{\norm{z_j}} \times \frac{z_k}{\norm{z_k}}\right)$ of two non-zero vectors $z_j$ and $z_k$ in $\mathbb{R}^3$. Then the Jacobian formulas of the cross product with respect to $z_j$ and $z_k$ can be calculated, respectively, as
\begin{align}
\frac{\partial}{\partial z_j} \left(\frac{z_j}{\norm{z_j}} \times \frac{z_k}{\norm{z_k}}\right)
=-\xi_{\frac{z_k}{\norm{z_k}}}\frac{\partial}{\partial z_j}\left(\frac{z_j}{\norm{z_j}}\right)
\end{align}
and
\begin{align}
\frac{\partial}{\partial z_k} \left(\frac{z_j}{\norm{z_j}} \times \frac{z_k}{\norm{z_k}}\right)
=\xi_{\frac{z_j}{\norm{z_j}}} \frac{\partial}{\partial z_k}\left(\frac{z_k}{\norm{z_k}}\right).
\end{align}
where $\xi_x$ denotes the cross product matrix given by
$\xi_x= \begin{bmatrix}
     		0 & -x_3 & x_2 \\
     		x_3 & 0 & -x_1 \\
     		-x_2 & x_1 & 0
  \end{bmatrix}$
for a vector $x= \begin{bmatrix}x_1 & x_2&x_3\end{bmatrix}^\top\in\mathbb{R}^{3}$.
\end{lemma}
\begin{proof} 
First note that the cross product $\left(\frac{z_j}{\norm{z_j}} \times \frac{z_k}{\norm{z_k}}\right)$ can be expressed with the cross product matrix as follows
\begin{align}
\frac{z_j}{\norm{z_j}} \times \frac{z_k}{\norm{z_k}}
=-\xi_{\frac{z_k}{\norm{z_k}}}\frac{z_j}{\norm{z_j}}
\end{align}
or
\begin{align}\frac{z_j}{\norm{z_j}} \times \frac{z_k}{\norm{z_k}}
=\xi_{\frac{z_j}{\norm{z_j}}} \frac{z_k}{\norm{z_k}}.
\end{align}
Thus, we can have the following results.
\begin{align}
\frac{\partial}{\partial z_j} \left(\frac{z_j}{\norm{z_j}} \times \frac{z_k}{\norm{z_k}}\right)
&=-\frac{\partial}{\partial z_j}\left(\xi_{\frac{z_k}{\norm{z_k}}}\frac{z_j}{\norm{z_j}}\right)\nonumber\\
&=-\xi_{\frac{z_k}{\norm{z_k}}}\frac{\partial}{\partial z_j}\left(\frac{z_j}{\norm{z_j}}\right)
\end{align}
and
\begin{align}
\frac{\partial}{\partial z_k} \left(\frac{z_j}{\norm{z_j}} \times \frac{z_k}{\norm{z_k}}\right)
&=\frac{\partial}{\partial z_k}\left(\xi_{\frac{z_j}{\norm{z_j}}} \frac{z_k}{\norm{z_k}}\right)\nonumber\\
&=\xi_{\frac{z_j}{\norm{z_j}}} \frac{\partial}{\partial z_k}\left(\frac{z_k}{\norm{z_k}}\right).
\end{align}
The proof is completed.
\end{proof}

\addtolength{\textheight}{-12cm}   
\bibliographystyle{IEEEtran}
\bibliography{rigidity2018}
\end{document}